\documentclass[12pt, draftclsnofoot, onecolumn]{IEEEtran}

\usepackage{float}
\usepackage{stackengine}
\usepackage{setspace}
\usepackage{amssymb,amsfonts,latexsym}
\usepackage{graphicx}
\usepackage{subfigure}
\usepackage{psfrag}
\usepackage{algorithmic}
\usepackage{blindtext}
\usepackage{comment}
\usepackage{cite}
\usepackage{url}
\usepackage{times}
\usepackage{bbm}
\usepackage{epsfig}
\usepackage{multirow}
\usepackage{longtable}
\usepackage{amsfonts}
\usepackage{amssymb}%
\usepackage{color}
\usepackage{mathrsfs}
\usepackage{bm}
\usepackage{booktabs}
\usepackage{threeparttable}
\usepackage[cmex10]{amsmath}
\usepackage{amsmath}
\usepackage[english]{babel}
\usepackage{mathbbold}
\usepackage{mathtools, nccmath}
\usepackage[linesnumbered,ruled]{algorithm2e}
\makeatother
\usepackage[table]{xcolor}
\usepackage{makecell}
\usepackage{stackengine}

\usepackage{stackengine}

\usepackage{stfloats}
\newtheorem{assumption}{\textbf{Assumption}}
\newtheorem{theorem}{\textbf{Theorem}}
\newtheorem{proof}{\textbf{Proof}}
\newtheorem{definition}{\textbf{Definition}}

\newtheorem{corollary}{\textbf{Corollary}}
\newtheorem{remark}{Remark}
\newtheorem{lemma}{\textbf{Lemma}}

\ifCLASSINFOpdf
\else
\fi

\begin{document}
		\title{Fast mmwave Beam Alignment via\\ Correlated Bandit Learning}
	\author{\small Wen Wu,~\IEEEmembership{\small Student~Member,~IEEE,}
			        Nan~Cheng,~\IEEEmembership{\small Member,~IEEE,}
	        Ning Zhang,~\IEEEmembership{\small Senior~Member,~IEEE,}
	      \\  Peng~Yang,~\IEEEmembership{\small Member,~IEEE,}
	      Weihua~Zhuang,~\IEEEmembership{\small Fellow,~IEEE,} and
	        Xuemin~(Sherman)~Shen,~\IEEEmembership{\small Fellow,~IEEE}
	\thanks{W. Wu, N. Cheng, P. Yang, W. Zhuang, and X. Shen are with the Department of Electrical and Computer Engineering, University of Waterloo, 200 University Avenue West, Waterloo, ON N2L 3G1, Canada (e-mail:\{w77wu, n5cheng, p38yang, wzhuang, sshen\}@uwaterloo.ca). N. Zhang is with Department of Computing Sciences, Texas A\&M University at Corpus Christi, TX, USA.~Email: ning.zhang@tamucc.edu.}}	
\maketitle

\begin{abstract}
Beam alignment (BA) is to ensure the transmitter and receiver beams are accurately aligned to establish a reliable communication link in millimeter-wave (mmwave) systems. Existing BA methods search the entire beam space to identify the optimal transmit-receive beam pair, which incurs significant BA latency on the order of seconds in the worst case. In this paper, we develop a learning algorithm to reduce BA latency, namely \underline{H}ierarchical \underline{B}eam \underline{A}lignment (HBA) algorithm. We first formulate the BA problem as a stochastic multi-armed bandit problem with the objective to maximize the cumulative received signal strength within a certain period. The proposed algorithm takes advantage of the \emph{correlation structure} among beams such that the information from nearby beams is extracted to identify the optimal beam, instead of searching the entire beam space. Furthermore, the \emph{prior knowledge} on the channel fluctuation is incorporated in the proposed algorithm to further accelerate the BA process. Theoretical analysis indicates that the proposed algorithm is asymptotically optimal. Extensive simulation results demonstrate that the proposed algorithm can identify the optimal beam with a high probability and reduce the BA latency from hundreds of milliseconds to a few milliseconds in the \emph{multipath} channel, as compared to the existing BA method in IEEE 802.11ad.



%


	\vspace*{0mm}
	\begin{flushleft}
		\textbf{\it Index Terms} -- mmwave, beam alignment, correlation structure,  prior knowledge, multi-armed bandit.
	\end{flushleft}	
\end{abstract}

\section{Introduction}

The ever-increasing data traffic driven by various emerging data-hungry applications, such as high-definition mobile video streaming, cordless virtual reality gaming and wireless fiber-to-home access, has placed a growing strain on the creaking traditional cellular networks. Millimeter-wave (mmwave) communication is envisioned as the most promising technology to accommodate the skyrocketing data traffic through harnessing multi-GHz bandwidths. Multiple standardization efforts, such as IEEE 802.11ad \cite{zhou2017enhanced,Ad_standard} and ongoing IEEE 802.11ay \cite{ayMagazine2017,wu2017Wiopt}, and large-scale field-trials have paved the road for the commercialization of mmwave communications.

In mmwave communication systems, narrow directional beams are adopted at both the transmitter and receiver to compensate for the huge attenuation loss. Since beams are narrow, the communication is possible only when the transmitter and receiver beams are properly aligned \cite{qiao2016proactive}, as shown in Fig. \ref{fig:beam_alignment}. Beam alignment (BA) is such a process to identify the optimal transmit-receive beam pair which attains the maximum received signal strength (RSS). Beam misalignment can dramatically reduce the link budget and drop the throughput from multiple Gbps to a few hundred Mbps \cite{infocom2018efficient}. As a key process in mmwave communications, BA is of significance to achieve multi-gigabit wireless transmission. To identify the best beam pair, a naive exhaustive search method scans all the combinations of the transmitter and receiver beams, which results in significant BA latency. Yet, a low-latency BA process is imperative for practical mmwave systems to accommodate real-time applications. Moreover, in mobile scenarios, user mobility changes the beam direction and thus frequently invokes BA, which further exacerbates the latency. To accelerate the beam search, IEEE 802.11ad protocol decouples the BA process into two steps. Firstly, the transmitter starts with a quasi-omnidirectional beam and the receiver scans the beam space for the best receiver beam. Secondly, the transmitter scans the beam space for the best transmitter beam while keeping the receiver quasi-omnidirectional. Still, the existing BA method in IEEE 802.11ad may take up to seconds with a large number of candidate beams \cite{hassanieh2018fast}. To reduce BA latency, can we identify the optimal beam without searching the entire beam space? 

\begin{figure}[t]
	\renewcommand{\figurename}{Fig.}
	\centering
	\includegraphics[width=0.45\textwidth]{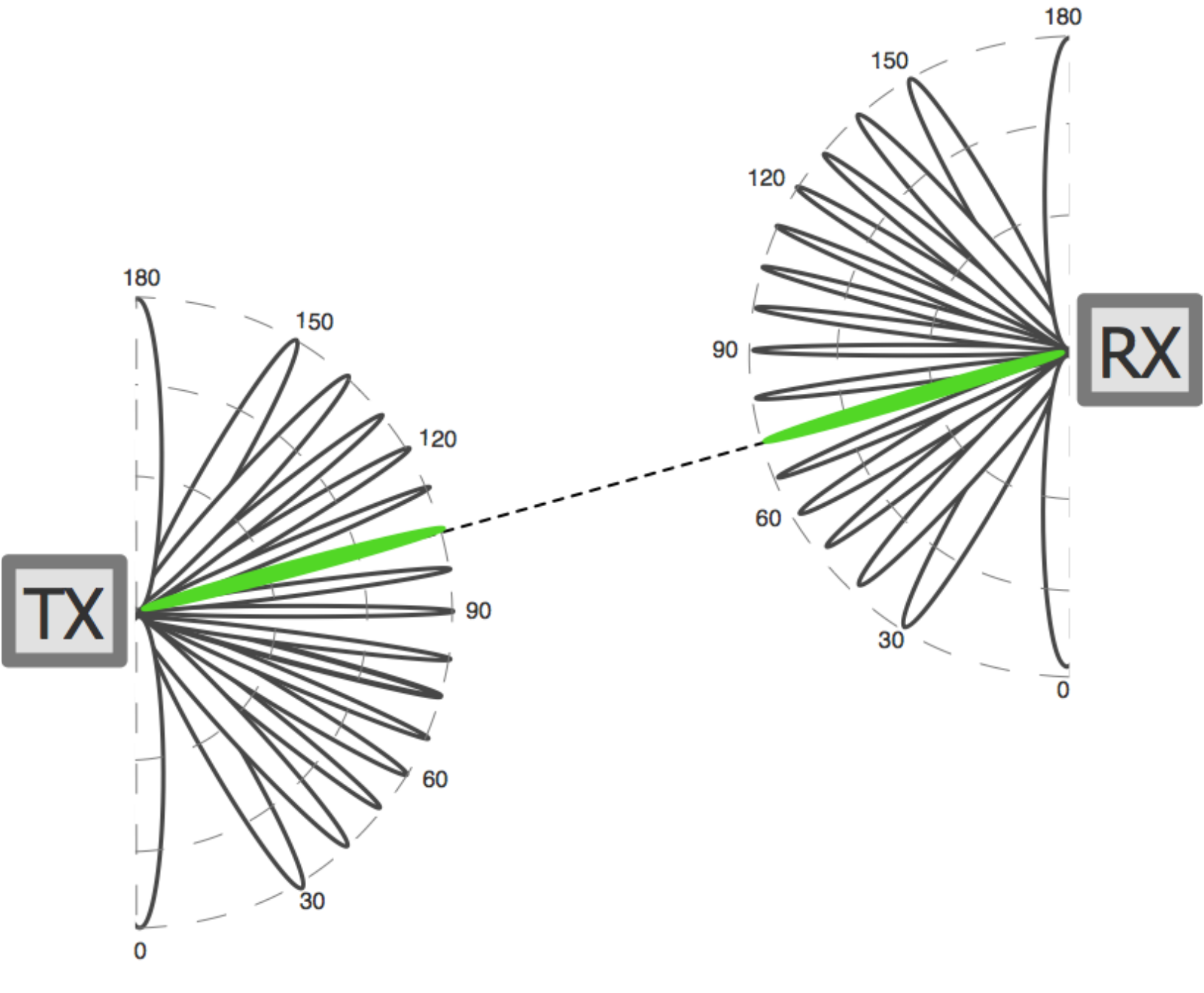}
	\caption{A beam alignment example with 16 beams. The well-aligned transmitter and receiver beams are represented by solid green beams.}
	\label{fig:beam_alignment}
	\vspace{-0.7cm}
\end{figure}


In the literature, there are some initial research efforts to address this challenge. Utilizing the sparse characteristic of the mmwave channel, Marzi \emph{et al.} developed a compressed sensing BA method  \cite{Marzi2016JSTSP}.  
Some out-of-band information, e.g., the Wi-Fi signal, is exploited to identify the optimal beam in \cite{sur2017wifi}. 
These works perform BA with the assistance of excessive extra information besides RSS. 
Surprisingly, a crucial feature, the correlation structure among beams, is ignored in previous works. In fact, the RSS of nearby beams is similar which means nearby beams are highly correlated. In this way, if a beam does not perform well, its nearby beams are highly likely to perform worse either. The measurement of one beam not only reveals information about itself, but also its nearby beams. Hence, the information from nearby beams can be learned to identify the optimal beam without searching the entire beam space. 




In this paper, we propose a fast BA algorithm, named Hierarchical Beam Alignment (HBA), by utilizing the \emph{correlation structure} among beams and the \emph{prior knowledge} on the channel fluctuation. In the BA problem, fast BA means identifying the optimal beam with the minimum latency. This problem boils down to sequentially selecting beams to maximize the cumulative RSS within a certain period, which can be formulated as a stochastic multi-armed bandit (MAB) problem. To solve this problem efficiently, two unique characteristics are incorporated in our proposed algorithm. Firstly, theoretical analysis indicates that the correlation structure among beams in the \emph{multipath} channel can be characterized by a \emph{multimodal} function. Utilizing this correlation structure, the proposed algorithm intelligently narrows the search space to identify the optimal beam. Secondly, incorporating the prior knowledge on the channel fluctuation to appropriately accommodate reward uncertainty, the proposed algorithm avoids excessive exploration and further accelerates the BA process. Theoretical analysis shows that the regret of HBA is \emph{bounded} and thus the proposed algorithm is asymptotically optimal. Extensive simulation results demonstrate that HBA can identify the optimal beam with a high probability and reduce the number of beam measurements in the multipath channel, even with coarse prior knowledge. Particularly, the proposed algorithm reduces the BA latency by orders of magnitude as compared to the BA method in IEEE 802.11ad. 

 Our contributions in this paper are summarized as follows.
\begin{itemize}
	\item We formulate the BA problem as a stochastic MAB problem, in which the objective is to sequentially select beams to maximize cumulative RSS within a certain period;
	\item We prove that the mean RSS function over the beam space follows a multimodality structure in the multipath channel, which characterizes the correlation structure among nearby beams;	  
	\item We propose a fast BA algorithm to accelerate beam search by exploiting the correlation structure and the prior knowledge on the channel fluctuation;
	\item We derive a sublinear analytical upper bound on the cumulative regret, i.e.,  $O(\sqrt{T\log T})$, indicating the proposed algorithm is asymptotically optimal.
\end{itemize}

The remainder of this paper is organized as follows. 
Section \ref{sec: related works} reviews related works. The system model and problem formulation are presented in Section \ref{sec: system model}. Section \ref{sec: unimodal bandit} proposes a fast BA algorithm. Section \ref{sec: regret_performance_analysis} analyzes the regret performance of the proposed algorithm. Simulation results are given in Section \ref{sec: simulation results}. Finally,  Section \ref{sec: conclusion} concludes this paper.

\section{Related Work}\label{sec: related works}
The BA problem in mmwave systems garners much attention recently. Zhou \emph{et al.} elaborated the challenges of the random access protocol in the BA process in dense networks \cite{zhou2017enhanced}. In addition, the authors developed possible solutions from the MAC perspective. Utilizing the sparse characteristic that only a few paths exist in the mmwave channel, a compressed sensing solution can align beams with a low beam measurement complexity of $O(L\log N)$, where $L$ is the number of channel paths and $N$ is the number of beams \cite{Marzi2016JSTSP}. The approach suits for mmwave systems where the accurate phase information is available. In another line of research, Wang \emph{et al.} developed a fast-discovery multi-resolution beam search in \cite{wang2009beam}, which probes the wide beam first and continues to narrow beams until identifying the best beam. While feasible, the method needs to adjust the beam resolution at every step. On the other hand, Xiao \emph{et al.}  proposed a hierarchical codebook search method to efficiently identify the optimal beam by jointly utilizing sub-array and deactivation techniques \cite{xiao2016hierarchical}. Moreover, they provide the closed-form expression of the hierarchical codebook. Sun \emph{et al.} further developed an orthogonal pilot based low-overhead beam alignment method for the multiuser mmwave systems \cite{sun2019beam}. 
Another solution exploits some out-of-band information, i.e., the Wi-Fi signal, to identify the optimal beam \cite{sur2017wifi}. {Similar works extract spatial information from sub-6 GHz signals to assist BA as well as boost throughput \cite{ali2018millimeter,hashemi2018out}.} 
Recent efforts leverage the multi-armed beams capability to improve BA performance. Hassanieh \emph{et al.}  proposed a fast BA protocol through scanning multiple directions simultaneously \cite{hassanieh2018fast}. A similar method, which treats the problem of identifying the optimal beam as that of locating the error in linear block codes, is developed to reduce BA complexity \cite{shabara2018linear}. The works in \cite{zhou2017enhanced, Marzi2016JSTSP,wang2009beam,sur2017wifi,hassanieh2018fast, ali2018millimeter,xiao2016hierarchical, hashemi2018out,sun2019beam, shabara2018linear} provide possible solutions for the BA problem in various scenarios. Different from prior works, our work considers the correlation structure among nearby beams to assist BA process. 

MAB theory has been widely applied in wireless networks, such as power allocation in small base stations \cite{wang2017small}\cite{shen2018generalized}, content placement in edge caching \cite{yang2018content,muller2017context}, task assignment in mobile crowdsourcing \cite{Yang2017bj} and mobility management in mobile edge computing \cite{sun2017emm}. Very recently, the BA problem is studied based on MAB theory, which makes online decision to strike the balance between \emph{exploitation} and \emph{exploration}. Gulati \emph{et al.} applied the celebrated upper confidence bound (UCB) algorithm in beam selection in traditional MIMO systems \cite{gulati2014learning}. 
Sim \emph{et al.}  developed an online beam selection algorithm in mmwave vehicular networks based on contextual bandit theory \cite{sim2018online}. This work learns information from real-time environment to enhance the throughput of mmwave networks.  
A pioneering work in \cite{infocom2018efficient} exploits a unimodal structure among beams to accelerate the BA process in static environments. This solution focuses on aligning beams in the single-path channel. Another work developed a distributed BA search method based on adversarial bandit theory \cite{chafaa2019adversarial}. These works provide highly relevant insights on the BA problem in mmwave networks via bandit learning theory. However, they do not provide a method to quickly and accurately align beams, especially in complicated multipath channels. Different from existing works, we focus on leveraging the correlation structure and prior knowledge to accelerate the BA process in the multipath channel with only RSS. 

\section{System Model and Problem Formulation}\label{sec: system model}

\subsection{Beam Alignment Model}


{As shown in Fig. \ref{fig:mmwave_structure}, we consider a point-to-point mmwave system in a static environment, where the transmitter is equipped with $N$ antennas. }Uniform linear arrays are assumed in both the transmitter and receiver, and each antenna element is connected to a phase shifter to form narrow directional beams \cite{wu2017performance}. In the BA process, the receiver keeps quasi-omnidirectional while the transmitter scans the beam space to identify the best one.  
We consider the sparse clustered channel model, i.e., Saleh-Valenzuela model \cite{Akdeniz2013Millimeter}. Suppose that the channel consists of $L$ paths: one dominant line-of-sight (LOS) path and $L-1$ non-line-of-sight (NLOS) paths, due to strong reflections from the ground or side walls.  
The channel array response between the transmitter and receiver can be represented as a mixture of sinusoids, 
\begin{equation}
\begin{split}
h_{n}=g_0e^{j\frac{2\pi d}{\lambda}n \vartheta_{0}}+ \sum_{l=1}^{L-1}g_l e^{j\frac{2\pi d}{\lambda}n \vartheta_{l}}
\end{split}
\end{equation} 
where $0\leq n\leq N-1$. Let $d$ and $\lambda$ denote the array element spacing and carrier wavelength, respectively. Typically, $d={\lambda}/{2}$. Let $g_0$ and $g_l$ represent the channel gains of the LOS path and the $l$-th NLOS path, respectively. Note that the channel gain of the NLOS path is around 10 dB weaker than that of the LOS path \cite{maltsev2009experimental}. Let $\theta$ denote the physical angle of the channel. The corresponding spatial angle of the channel is denoted by $\vartheta=\cos \theta_{}$. We vectorize the sinusoids $e^{j {2\pi d}n \vartheta_{}/{\lambda}}, 0 \leq n\leq N-1$ into a vector $\mathbf{x}(\vartheta_{})\in \mathbb{C}^{N\times 1}$.  
 Thus, the channel vector is given by
 \begin{equation}
 \mathbf{h}=g_0  \mathbf{x}(\vartheta_{0})+\sum_{l=1}^{L-1}g_l  \mathbf{x}(\vartheta_{l})\in \mathbb{C}^{N \times 1}.
 \end{equation}
{Since we consider a static environment, the channel vector keeps invariant during the BA process. }

\begin{figure}[t]
	\renewcommand{\figurename}{Fig.}
	\centering
	\includegraphics[width=0.45\textwidth]{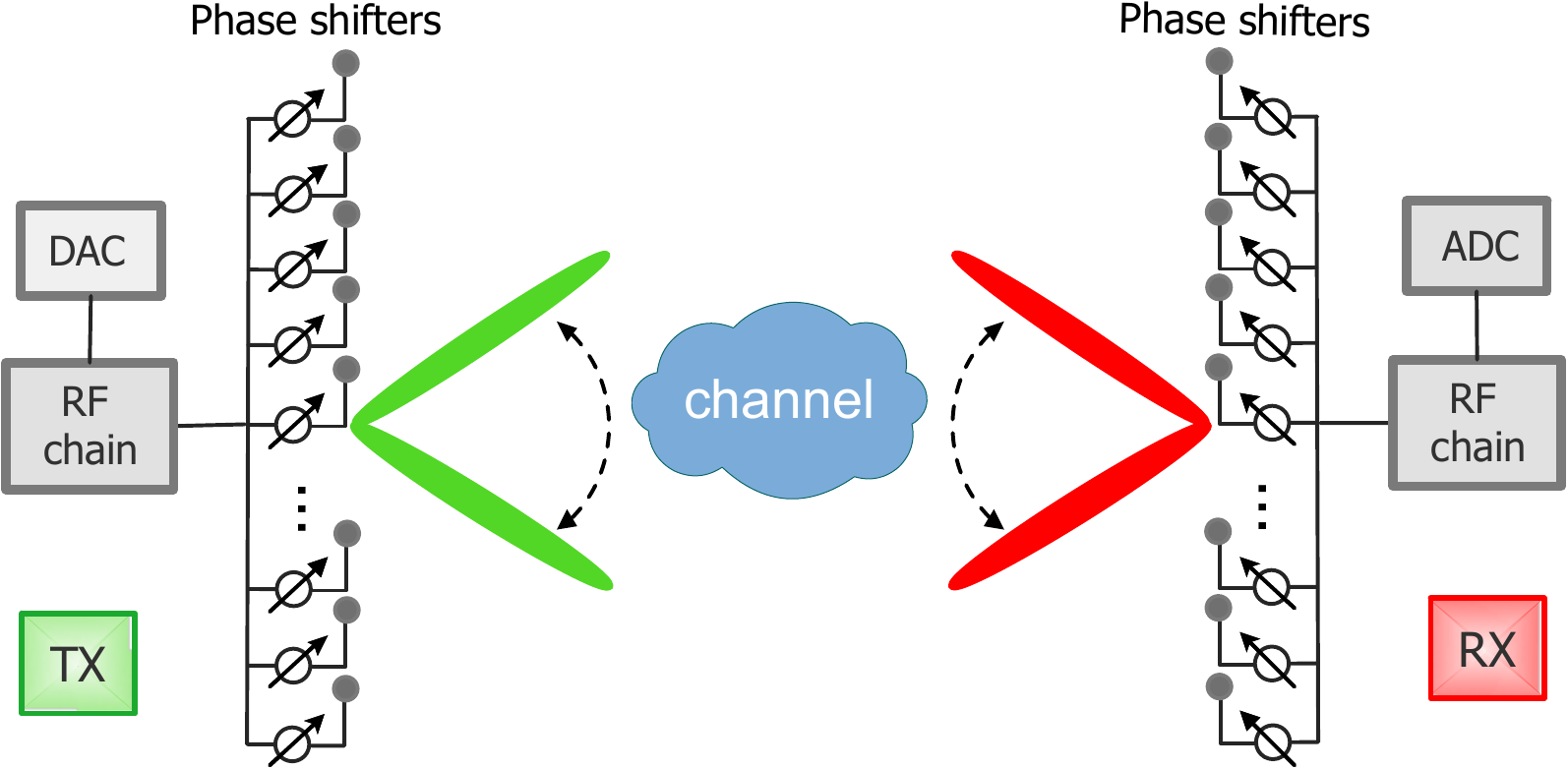}
	\caption{The point-to-point mmwave system.}
	\label{fig:mmwave_structure}
	\vspace{-0.7cm}
\end{figure}

Let $\mathbf{W}=[\mathbf{w}_1,\mathbf{w}_2,...,\mathbf{w}_N]\in \mathbb{C}^{N \times N}$ denote the unitary Discrete Fourier Transform (DFT) matrix whose columns constitute the transmit beam space, given by
\begin{equation}\label{equ:W}
\mathbf{W}=\frac{1}{\sqrt{N}}[\mathbf{x}(\omega_1), \mathbf{x}(\omega_2),..., \mathbf{x} (\omega_{N})].
\end{equation}
In \eqref{equ:W}, $\omega_i=\frac{2i-N}{N}$ represents the spatial angle of the $i$-th beam \cite{Marzi2016JSTSP}. 
According to the BA method in IEEE 802.11ad, the transmitter scans all the beams in $\mathbf{W}$, while the receiver beam keeps omni-directional. The received signal vector is given by  
\begin{equation}\label{equ:received_RSS}
\mathbf{y}=\sqrt{P} \mathbf{h}^H\mathbf{W}+\mathbf{n}
\end{equation}
where $\mathbf{n}$ denotes the additive Gaussian white noise vector. Let $N_oW$ denote the mean noise power, where $W$ is the channel bandwidth and $N_o$ is the noise power density.

The problem of identifying the optimal transmit beam boils down to identifying the element with the maximum magnitude within $\mathbf{y}$. Hence, to identify the optimal beam, the BA method in IEEE 802.11ad protocol needs to measure the RSS of all the transmit beams, leading to a high beam measurement complexity \cite{hassanieh2018fast}. Searching the entire beam space incurs significant BA latency, especially when the beam space is large.

%






\subsection{Problem Formulation}
In this subsection, the BA problem is formulated as a stochastic MAB problem for \emph{stationary} environments. Consider a time slotted system with $T$ time slots of equal duration. In time slot $t\in \{1,2,..., T\}$, the transmitter selects a beam to transmit data. Let $\mathcal{B}=\{b_1, b_2,..., b_N\}$ denote the set of candidate beams, which can be considered as \textit{arms} in the bandit theory. 
At the beginning of time slot $t$, the transmitter selects a beam denoted by $b^t \in \mathcal{B}$. At the end of time slot $t$, the transmitter observes noisy RSS from the receiver, i.e., $r\left(b^t\right)$, which is considered as a \textit{reward}.  
Rigorously, the reward is a random variable due to the channel fluctuation, such as shadow fading and the disturbance effect. For simplicity, we assume that the reward follows a Gaussian distribution with a variance $\sigma^2$. In other words, $\sigma^2$ also represents the variance of the channel fluctuation, which is utilized as \emph{prior knowledge} in the following algorithm design. 
Note that the proposed algorithm can also be applied to non-Gaussian distribution settings, as validated in Section \ref{sec: simulation results}.


Let $b^{1 : t}=\{b^1, b^2,..., b^t\}$ denote the sequentially selected beams up to time slot $t$. The set of corresponding sequential rewards is represented by $r^{1:t}=\{r(b^1), r(b^2),..., r(b^t)\}$. In the MAB setting, a sequential beam selection \emph{policy} is how the transmitter selects the next beam based on previously selected beams $b^{1:t}$ and observed rewards $r^{1:t}$. Let $\Pi$ be the set of all possible sequential beam selection policies. Our objective is to find a policy, $\pi \in \Pi$, that maximizes the expected cumulative reward (RSS) within a given time horizon of $T$ slots, i.e., $\sum_{t=1}^{T}r(b^t)$. This objective conforms our target since a fast BA algorithm is to identify the optimal beam with the minimum latency.

In the MAB theory, \emph{expected cumulative regret} is commonly adopted to evaluate the performance of a given policy, which denotes the expected cumulative difference between the reward of the selected beam and the maximum reward achieved by the optimal beam. The \emph{expected cumulative regret} is defined as
\begin{equation}
\begin{split}
R^{\pi}(T)
&=\mathbb{E}\left [\sum_{t=1}^{T}\left( {r}(b{^\star}) -r(b^t)\right) \right]
=T\cdot \mathbb{E}\left [ r\left(b{^\star}\right)\right]-\sum_{b_i\in \mathcal{B}}^{} N_{b_i}^{\pi}(T)\mathbb{E}\left [r\left(b_i\right)\right]
\end{split}
\end{equation}
where $b{^\star}$ represents the optimal beam and $N_{b_i}^{\pi}(T)$ denotes the number of times that $b_i$ has been selected up to time slot $T$. Hence, maximizing the cumulative reward is equivalent to minimizing the \emph{expected cumulative regret} within $T$ \cite{infocom2018efficient}, which can be expressed as


%


\begin{subequations}\label{problem1}
	\begin{align}
	{\mathcal{P}1:}	& \underset{\pi \in {\Pi}}{\text{min}}
	& & R^{\pi}(T) \nonumber \\
	& \text{s.t.}  
	& &  \sum_{b_i\in \mathcal{B}}^{} N_{b_i}^{\pi}(T)\leq  T\\ \label{problem1, constaint 1}
	& &&N_{b_i}^{\pi}(T)\in \mathbb{Z}, \forall b_i\in \mathcal{B}. 
	\end{align}
\end{subequations}

The preceding MAB problem $\mathcal{P}1$ can be solved by the celebrated UCB algorithm \cite{auer2002finite}. However, this problem has two characteristics that were not utilized in the UCB algorithm. Firstly, since the RSS of nearby beams are highly correlated, the correlation information from nearby beams can be utilized to select the next beam efficiently. Secondly, the prior knowledge on the channel fluctuation reflects the information of environment, which can be exploited to appropriately accommodate reward uncertainty such that the BA process can be further accelerated. 
In the following, we will leverage these two characteristics to accelerate the convergence speed, and hence reduce BA latency.

\section{Fast Beam Alignment}\label{sec: unimodal bandit}
In this section, we first analyze and validate that the mean reward (RSS) over the beam space follows a multimodality structure, which characterizes the inherent correlation among beams. Next, by exploiting the correlation structure and the prior knowledge, a fast BA algorithm is proposed to identify the optimal beam. 

\subsection{Correlation Structure}


Consider a \emph{cyclic} undirected graph ${G}=({\mathcal{B}},{{E}})$ whose vertices ${\mathcal{B}}$ stand for the beams. Let $(b_i,b_{i+1})\in E$ denote the edge that connects neighboring beams $b_{i}$ and $b_{i+1}$. 
In addition, $(b_N,b_1)\in E$ indicates that the last beam $b_N$ and the first beam $b_1$ are neighbors since their beam orientations are close to each other. 
The unimodality structure is defined as follows. 
\begin{definition}\label{definition:unimodal}
	\textbf{(Unimodality)} Let $b_{i^\star}$ denote the optimal beam in ${G}$. The \emph{unimodality} structure indicates that, $\forall b_i \in \mathcal{B}$, there exist a path, $(b_{i},b_{i+1},...,b_{i^\star})$, along which the mean reward is strictly increasing. 


\end{definition}

	In other words, the unimodality structure means that there is no local optimal beam over the beam space. Next, we aim to show that the correlation structure among beams follows above unimodality structure. Consider the single-path channel, where $g$ and $\vartheta$ represent the channel gain and channel spatial angle of the path, respectively. With \eqref{equ:received_RSS}, the mean RSS is given by 

\begin{equation}\label{equ:antenna_directivity_gain}
\begin{split}
\mathbb{E}\left[	r({b_i})\right]
&=P\left| \mathbf{h}^H\mathbf{w}_i\right|^2+N_oW\\
&=\frac{Pg^2 }{N }\left|\mathbf{x}^H(\vartheta) \mathbf{x}(\omega_{i})\right|^2+N_oW\\
&=\frac{Pg^2}{N }\left| \sum_{n=0}^{N-1}e^{j\frac{2\pi d}{\lambda}n (\omega_i-\vartheta_{})} \right|^2+N_oW\\
&=\frac{Pg^2}{N }D\left(\omega_i-\vartheta_{}\right)+N_oW, \forall b_i\in \mathcal{B}
\end{split}
\end{equation}	
where 
\begin{equation}
D(x)=\frac{\sin^2({N  \pi d x}/{\lambda})}{ \sin^2({\pi d x}/{\lambda})}
\end{equation}
denotes the antenna directivity function, which depends on the angular misalignment $x$. Hence, the mean RSS is a function of angular misalignment $\omega_i-\vartheta$.

\begin{theorem}\label{lemma:unimodality_reward_function}
In the single-path channel, the mean reward (RSS)  over the beam space is a unimodal function. 
\end{theorem}
\begin{proof}
Proof is provided in Appendix \ref{appendix:theorem multimodality}.
\end{proof}

The linear combination of several unimodal functions is a \emph{multimodal} function, which means that there exist several local optimums.
\begin{corollary}\label{proposition:multimodal}
	 In the multipath channel, the mean reward (RSS) over the beam space is a multimodal function. The dominant peak of the multimodal function is caused by the LOS path, while other peaks are caused by NLOS paths.
\end{corollary}
\begin{proof}
	Proof is provided in Appendix \ref{appendix:tmultimodality}.
\end{proof}

For example, Fig. \ref{Fig:Graphical unimodality} shows the RSS function over the beam space in a two-path channel. Even though the practical RSS is noisy due to the channel fluctuation, we observe that the mean RSS function follows the multimodality structure. For a two-path mmwave channel, there exists two peaks in the mean RSS function, where the dominant peak is due to the LOS path and another smaller peak is due to the NLOS path. Furthermore, the multimodality structure has been observed in many in-field measurements in mmwave systems, which further validates our theoretical results. 

  \begin{figure}[t]
	\centering
		\vspace{-0.5cm}
	\renewcommand{\figurename}{Fig.}
	\includegraphics[width=0.45\textwidth]{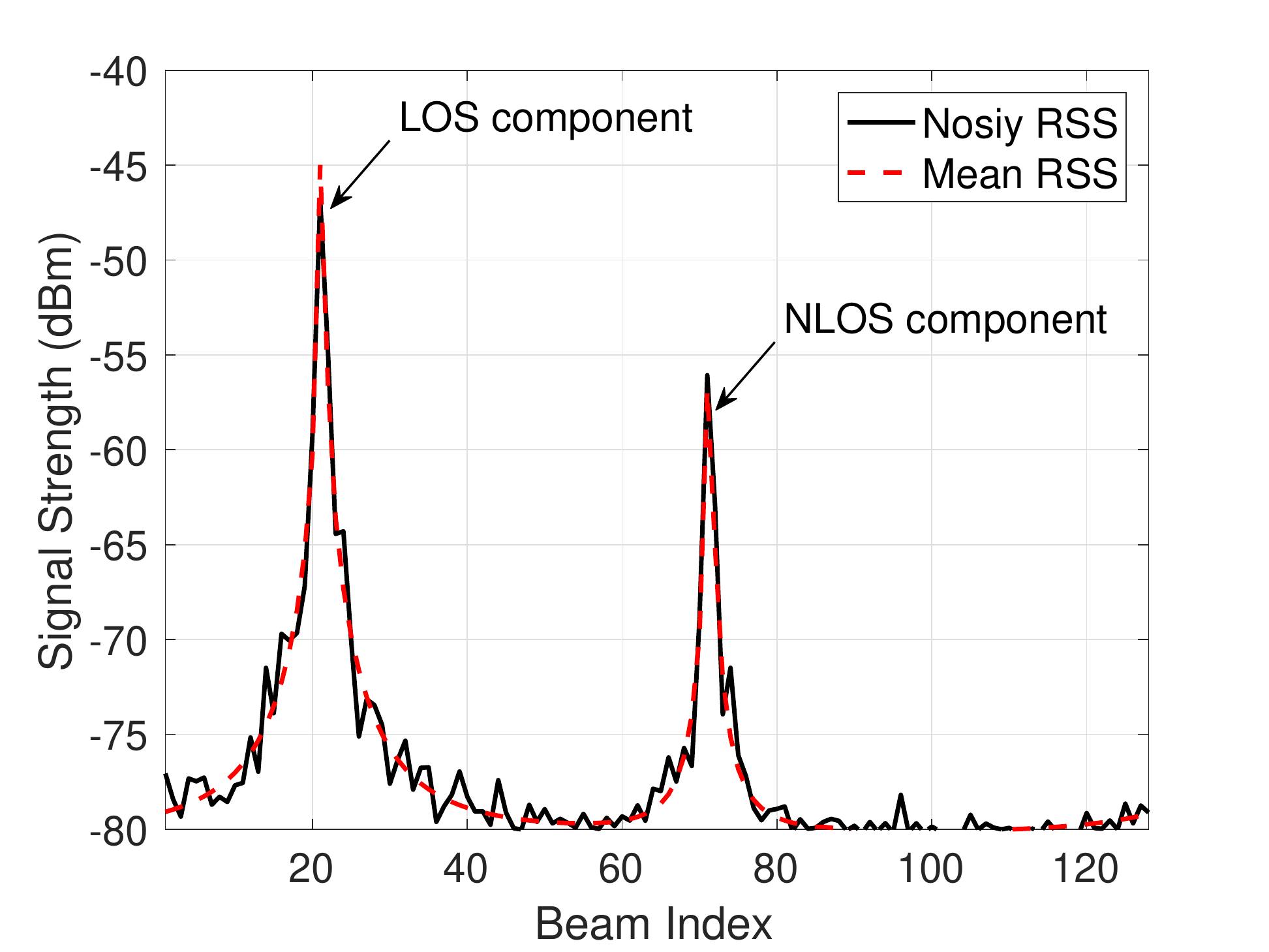}
	\vspace{-0.3cm}
	\caption{The RSS function over the beam space in a two-path channel with 128 beams. The peak caused by the LOS link is around 10 dB higher than that by the NLOS link.}
	\label{Fig:Graphical unimodality}
	\vspace{-0.7cm}
\end{figure}

\begin{remark}
	Theoretical analysis indicates that the RSS depends on the angular misalignment. As the angular misalignment of nearby beams is close, the RSS of nearby beams is similar such that nearby beams are highly correlated. 
	
	
\end{remark}

{Since the RSS function over the beam space is a multimodal function, the BA problem boils down to identifying the optimal point of a multimodal function. In other words, our goal is to find  optimal point $x^\star$ that maximizes multimodal reward function $f(x), x\in \mathcal{X}$. To solve this problem efficiently, the correlation structure of the reward function is leveraged. Specifically, the correlation structure is exploited based on a \emph{dissimilarity function} that captures the smoothness of  reward function \cite{NIPS}. }
	\begin{definition}
	{ \textbf{Dissimilarity}. For space $\mathcal{X}$, a dissimilarity function for $x_1\in \mathcal{X}$ and $x_2\in \mathcal{X}$ is defined as $q(x_1,x_2)=w\|x_1-x_2\|^\beta$, where $w>0$, $\beta>0$ and $\| \cdot\|$ denotes the Euclidean norm function. Note that $q(x,x)=0$ for $x\in \mathcal{X}$. }
	\end{definition}

	{The dissimilarity function is applied to characterize the discrepancy of two points in the reward function. Normally, two nearby points in the function have similar rewards, which means the dissimilarity between two nearby points is bounded. Such smoothness property of the reward function is exploited in the following algorithm design to accelerate the BA process. }

\subsection{Prior Knowledge}
In addition to the aforementioned correlation structure, some prior knowledge can be leveraged to further speed up the BA process. As the reward is impacted by wireless environments, channel fluctuation statistics reflects the underlying information of the wireless environments. Leveraging the channel fluctuation statistics can appropriately accommodate the reward uncertainty such that less exploration is required. Specifically, the variance of the channel fluctuation $\sigma^2$ is assumed to be known \emph{a priori} to accelerate the BA process. In practice, the prior knowledge can be obtained in the system initialization phase before the BA process is invoked. Practical mmwave systems also collect the variance of channel fluctuation periodically. Besides, since the channel statistical information changes slowly in static environments, there is no
need to frequently collect the information. It is worth noting that the proposed algorithm works even with coarse prior knowledge at the expense of slower convergence or lower beam detection accuracy, which is presented in Section \ref{sec: simulation results}.

\subsection{Hierarchical Beam Alignment (HBA) Algorithm}\label{sec: proposed algorithm}

As discussed, the mean reward function exhibits the multimodality structure, and hence we adapt and extend the hierarchical optimistic optimization (HOO) algorithm \cite{NIPS} to the BA problem. Due to the lack of prior knowledge, HOO adopts a large confidence margin to accommodate the reward uncertainty, which results in slow convergence. Similar to the well-known Bayesian principles in \cite{reverdy2014modeling}, we leverage the prior knowledge to obtain an appropriate confidence margin, which avoids unnecessary exploration and further accelerates convergence. The proposed HBA algorithm is sketched in Algorithm \ref{Alg:HBA scheme}. {In the algorithm, $Ber(0.5)$ represents a Bernoulli distributed random variable with a parameter of 0.5, which means that the random variable is equally likely to take values 0 and 1. In addition, $leaf(\mathcal{T})$ represents the leaf node of tree $\mathcal{T}$.}

The proposed algorithm is designed based on the correlation structure among beams. If a beam performs well, its nearby beams are highly likely to perform well too. {Hence, the core idea is to explore intensively around good beams while loosely in others.} For this purpose, a search tree is constructed, whose nodes are associated with search regions. A deeper node represents a smaller search region, as an illustrative example shown in Fig. \ref{fig:zooming_process}. The algorithm operates in discrete time slots, and the binary tree is constructed in an incremental manner. At each time slot, a new node is selected by a node selection process and added to the search tree.  {Once selected, the beam located in the selected node is measured, and then the corresponding reward is observed}. Then, the attributes of the search tree are updated based on the newly observed reward. In this way, the algorithm intelligently narrows the search region until the optimal beam is identified. It is worth noting that selecting a new node means exploring the region associated to the node, and the search tree explores the region based on previously selected beams and observed rewards.  


\begin{algorithm}[t]\label{Alg:HBA scheme}
	\caption{HBA algorithm}
	\LinesNumbered   
	\SetKwInOut{Input}{Input}
	\Input{ ${\zeta}$, $\rho_1$, $\gamma$ and $\sigma^2$ }
	\SetKwInOut{Output}{Output}
	\Output{$b^\star$}
	Initialization: Set $\mathcal{T}=\{(0,1)\}$, $Q_{2,1}=Q_{2,2}=+\infty$, $x_L=0$ and $x_H=1$\; 
	\For{\text{t=1,2,3...} }
	{
		$(h,j) \leftarrow (0,1)$, $\mathcal{P}\leftarrow \{(h,j)\}$\;
		{$\rhd$ New node selection}\\
		\While{$(h,i)\in \mathcal{T}_t$ }
		{		
			\begin{algorithmic} 
				\IF{ $Q_{h+1,2j-1}\left(t\right)>Q_{h+1,2j}\left(t\right)$}
				\STATE  $(h,j)\leftarrow (h+1,2j-1)$, update ${x_L}=x_a$\;
				\ELSIF{$Q_{h+1,2j-1}\left(t\right) <Q_{h+1,2j}\left(t\right)$}
				\STATE   $(h,j)\leftarrow (h+1,2j)$, update ${x_H}=x_a$\; 
				\ELSE \STATE  $(h,j)\leftarrow (h+1,2j-$$Ber(0.5))$,
				 update the search region\;
				
				\ENDIF
			\end{algorithmic}
			$\mathcal{P}\leftarrow \mathcal{P}\cup \{(h,j)\}$\;
		}
		$(H_t,J_t)\leftarrow (h,j)$; $\mathcal{T}_{t+1}= \mathcal{T}_t\cup \{(H_t,J_t)\}$;\\

		 {$\rhd$ Attributes update}\\
		{Measure the beam located in the center $C_{H_t,J_t}$, and observe the reward $r^{t}$\;}
		 
 $\forall (h,j)\in \mathcal{P}$, update $N_{h,j}\left(t\right)$ and $R_{h,j}\left(t\right)$ with \eqref{equ:N-values} and \eqref{equ: R_value}, respectively;\\
 $ \forall (h,j)\in \mathcal{T}_t$, update $E_{h,j}\left(t\right)$ with \eqref{equ:E_value};\\
 		$Q_{H+1,2J-1}\left(t\right)= Q_{H+1,2J}\left(t\right)= +\infty$; $\mathcal{\hat{T}}=\mathcal{T}_t$\;
		\For{\text{ $(h,j)\in \mathcal{\hat{T}}$} }
		{
			$(h,j)\leftarrow leaf(\mathcal{\hat{T}})$, update $Q_{h,j}\left(t\right)$ with \eqref{equ: Q_values update}, $\mathcal{\hat{T}} \leftarrow \mathcal{\hat{T}}\setminus (h,j) $;
		}	
	 {$\rhd$ Terminating condition}\\
	 	\begin{algorithmic} 
	 		\IF{$ x_H-x_L < {\zeta}/{N} $}			
	 			\STATE{Terminate beam search and select current beam $b^\star$}\;
	 		\ENDIF
	 	\end{algorithmic} 
	}
	
\end{algorithm}

Next, we elaborate the algorithm in detail. In the initialization phase, the beam space, $\mathcal{B}$, is mapped to a region $\mathcal{X}=[0,1]$, which is uniformly partitioned by each beam. Similarly, the RSS function, $r(b_i), \forall b_i\in \mathcal{B}$, is mapped to a normalized reward function, $f(x), \forall x\in \mathcal{X}$, within $[0,1]$. 
In the beginning, the search tree $\mathcal{T}$ only contains a root node $(0,1)$. The node in the tree is represented by $(h,j)$, where $h$ denotes the depth from the root node and $ j, 1\leq j \leq 2^h$ denotes the index at depth $h$. In addition, each node in the tree is associated with a region. Let ${C}_{h,j}$ represent the region of $(h,j)$. Specifically, the root node represents the entire region, i.e., ${C}_{0,1}=[0,1]$. Let $(h+1,2j-1)$ and $(h+1,2j)$ denote the left and the right child node of $(h,j)$, respectively. Two child nodes partition the region of their parent node. Consider ${C}_{h,j}=[{x}_L,x_H]$, the left child node is associated with a region ${C}_{h+1,2j-1}=[{x}_L,{x}_a]$ and the right child node is associated with a region ${C}_{h+1,2j}=[{x}_a,{x}_H]$, where $x_a=x_L+\left({x_H}-{x_L}\right)/{2}$ is the middle point of ${C}_{h,j}$. {The HBA algorithm operates in a ``zooming" manner, which intelligently narrows the search region via comparing the $Q$-values in the tree. The $Q$-value is designed based on the correlation structure of the reward function and the prior knowledge.}  
At time slot $t$, HBA consists of the following three phases: 

1. \emph{New node selection}. In this phase, a new node will be selected. Let $\mathcal{T}_t$ denote the tree at time $t$. {At each time slot, starting from the root node, the $Q$-values of two child nodes are compared until a new node $(H_t,J_t)\notin \mathcal{T}_t$ is selected. }
Specifically, traversing the tree, the child with a higher $Q$-value is chosen, otherwise breaking ties randomly (lines 5-6). The selected node is added to the tree, i.e., $\mathcal{T}_{t+1}= \mathcal{T}_t\cup \{(H_t,J_t)\}$, and the path from the root node to the selected node is stored in $\mathcal{P}$.  


2. \emph{Attributes update}. In this phase, the attributes of all the nodes in the tree are updated. {For the selected node in the previous phase, a beam located in the center of $C_{H_t,J_t}$ is measured and then the corresponding reward $r_{t}$ is obtained. Based on the newly observed reward, for node $(h,j)$, $Q_{h,j}$ is updated by the following steps}. 

Firstly, as the new node is the descendant of all the nodes in path $\mathcal{P}$, $N_{h,j}\left(t\right)$, which represents the number of times that ${(h,j)}$ has been selected until time slot $t$, is updated by 
\begin{equation}\label{equ:N-values}
N_{h,j}\left(t\right)= N_{h,j}\left(t-1\right)+1,\forall  (h,j)\in \mathcal{P}.
\end{equation}

Secondly, $R_{h,j}\left(t\right)$ represents the \emph{mean measured reward} of $(h,j)$ up to time slot $t$, which is updated by
\begin{equation}\label{equ: R_value}
R_{h,j}\left(t\right)=\frac{\left(N_{h,j}\left(t\right)-1\right)R_{h,j}\left(t-1\right)+r^{t}}{N_{h,j}\left(t\right)}
,\forall  (h,j)\in \mathcal{P}.
\end{equation}

{Thirdly, for each node in the tree, the \emph{initial estimated maximum mean reward} in region $C_{h,j}$, denoted by $E_{h,j}\left(t\right)$, is updated by, }
\begin{equation}\label{equ:E_value}
E_{h,j}\left(t\right)=
 \begin{cases}
 R_{h,j}\left(t\right)+\sqrt{\frac{2\sigma^2\log t}{N_{h,j}\left(t\right) }}+\rho_1\gamma^h ,        &\hfill   \text{if } N_{h,j}\left(t\right)>0 \\
+\infty, &\hfill \text{otherwise}
 \end{cases}
\end{equation}
where $\sqrt{\frac{2\sigma^2\log t}{N_{h,j}\left(t\right)}}$ is the confidence margin to accommodate for the uncertainty of rewards. As aforementioned, we adopt the Bayesian principle to design the confidence margin by leveraging the prior knowledge on the variance of channel fluctuation. {In \eqref{equ:E_value}, $\rho_1\gamma^h $ accounts for the maximum variation of the mean reward function in region $C_{h,j}$, where $\rho_1>0$ and $\gamma\in (0,1)$.} { This term is obtained via the correlation structure in the reward function. The maximum dissimilarity within region $C_{h,j}$ for the reward function is upper bounded by $\rho_1\gamma^h$, i.e., $\underset{x_1,x_2\in C_{h,j}}{\max} q(x_1,x_2)\leq \rho_1\gamma^h,\forall x_1,  x_2\in \mathcal{X}$, which holds due to the bounded diameter assumption in Section \ref{sec: regret_performance_analysis}.  The values of $\rho_1$ and $\gamma$ are selected based on extensive simulation trials. For a binary tree case, $\gamma$ is typically set to 0.5 \cite{NIPS}. Note that $E$-values of all the unexplored nodes are set to infinity.}

{Finally, the \emph{estimated maximum mean reward} in region $C_{h,j}$, $Q_{h,j}\left(t\right)$, should be recursively updated through the following bound}
\begin{equation}\label{equ: Q_values update}
Q_{h,j}\left(t\right)=
 \begin{cases}
 \min\{E_{h,j}\left(t\right),\max\{Q_{h+1,2j-1}\left(t\right),Q_{h+1,2j}\left(t\right)\}\}, &\hfill   \text{if } N_{h,j}\left(t\right)>0 \\
 +\infty, &\hfill \text{otherwise}
 \end{cases}
\end{equation}

This bound depends on two terms. The first term, $E_{h,j}\left(t\right)$, is an upper bound for $Q_{h,j}\left(t\right)$ due to the definition of  $E$-values. The second term, $\max\{Q_{h+1,2j-1}\left(t\right),Q_{h+1,2j}\left(t\right)\}$, is another valid upper bound of $Q_{h,j}\left(t\right)$. Since $C_{h,j}=C_{h+1,2j-1}\cup C_{h+1,2j-1}$, the maximum value between the $Q$-values in two subsets is the upper bound of $Q$-value in the union set. Combining both terms together, a tighter upper bound is obtained via taking the minimum value of these two bounds. 
Note that $Q$-values should be updated from the leaf node of the tree because $Q$-values of child nodes form the upper bound of their parent node (lines 12-15). 


3. \emph{Terminating condition}. As the tree is constructed over time, the search region gradually narrows as the depth of the tree increases. {When the search region is sufficiently small, i.e., $x_H-x_L<\zeta/N$ where $0<\zeta<1$, the BA process is terminated and the beam located in the final region is selected as the optimal beam.} The value of $\zeta$ should be carefully selected based on extensive simulation trials. Noteworthily, a larger $\zeta$ value results in faster convergence while lower beam detection accuracy.

\begin{remark}
	{A region attained a large $Q$-value represents that the potential maximum reward in the region is high, which means that the optimal beam (the maximum reward) locates in this region with a high probability. Hence, the HBA algorithm explores intensively in the regions with high estimated maximum rewards ($Q$-values) while loosely in others. In this way, the HBA algorithm is more efficient than the exhaustive search method, which accelerates the BA process.}
	
\end{remark}

\begin{figure}[t]
	\centering
	\renewcommand{\figurename}{Fig.}
	\begin{subfigure}[Zooming process]{
			\label{fig:zooming_process}
			\includegraphics[width=0.45\textwidth]{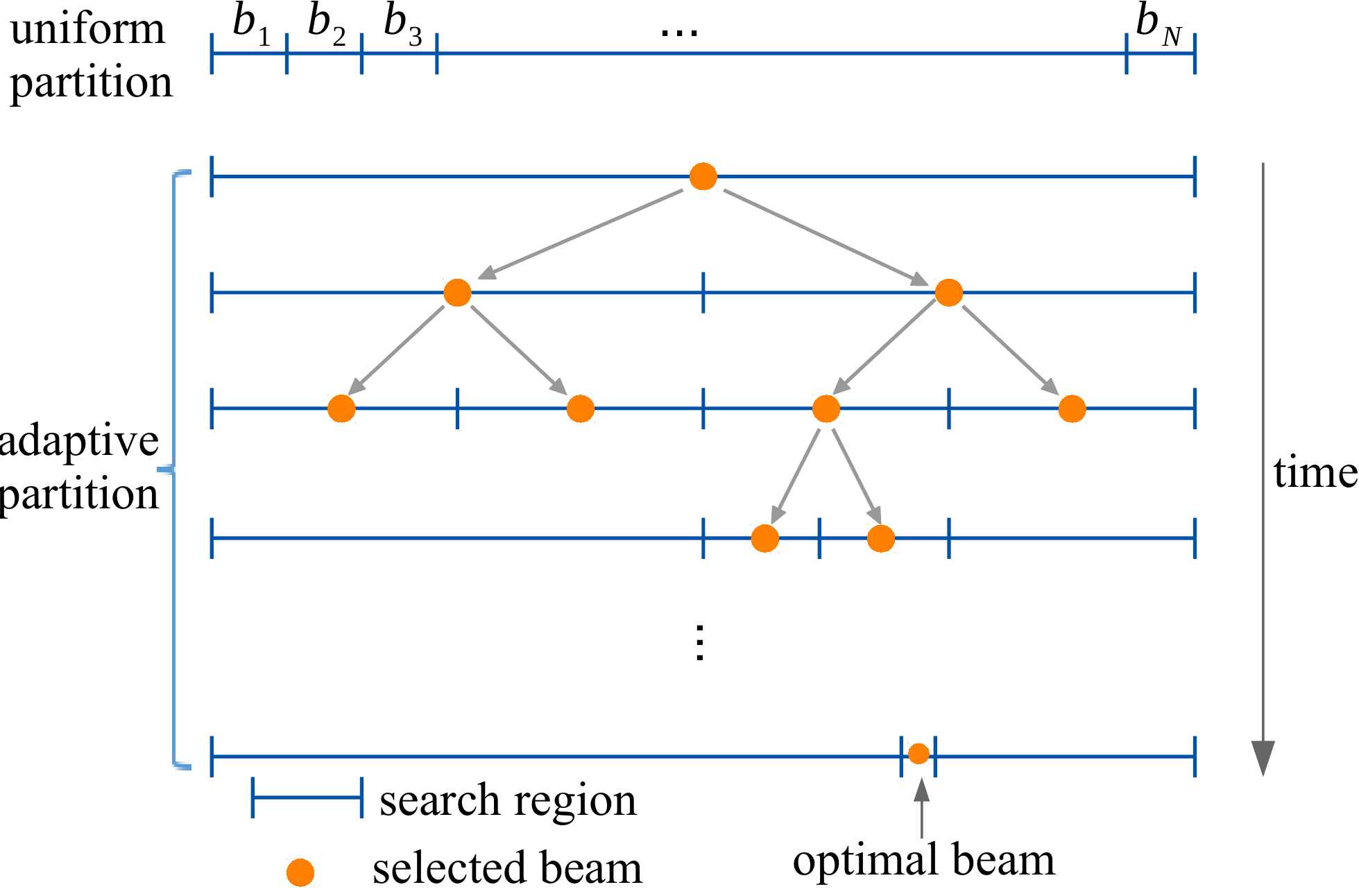}}
	\end{subfigure}%
	~
	\begin{subfigure}[Sequentially selected beams]{
			\label{fig:batch_measurement}
			\includegraphics[width=0.45\textwidth]{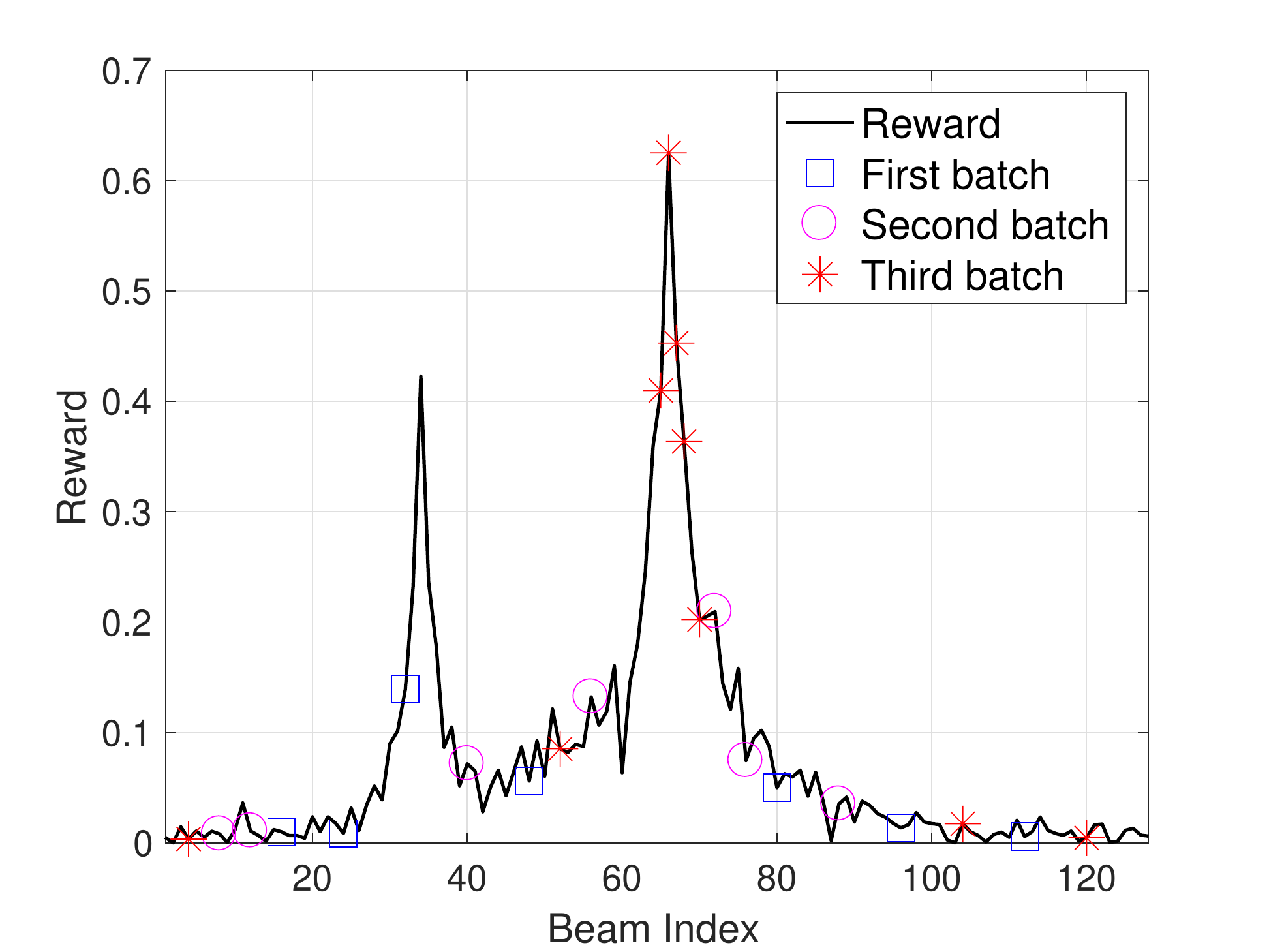}}
	\end{subfigure}
	\vspace{-0.3cm}
	\caption{Illustrative examples of the HBA algorithm. (a) The proposed algorithm operates in a ``zooming" manner. (b) The region that contains the dominant peak is explored intensively, while others are explored loosely.}
	\label{Fig:illustration}
		\vspace{-0.7cm}
\end{figure}

\textbf{Illustrative example}: For better understanding of HBA, we provide two illustrative examples in Fig. \ref{Fig:illustration}. Firstly, as shown in Fig. \ref{fig:zooming_process}, HBA operates similar to a ``zooming" process. At the beginning, the search region is the entire region, which is uniformly partitioned by the beams. As time goes by, the search region is adaptively partitioned, and the algorithm gradually zooms to the region that contains the optimal beam. 
Secondly, sequentially selected beams in the BA process are depicted in Fig. \ref{fig:batch_measurement}. The selected beams are divided into three batches according to the timeline. The first batch beams locate randomly in the whole region. The second batch beams get closer to the dominant peak. The last batch beams mainly focus around the optimal beam. We observe that the proposed algorithm explores intensively in the regions that contain good beams while loosely on the others.



\subsection{Complexity Analysis}
At time slot $T$, $\mathcal{T}_t$ contains $T$ nodes as the tree increments by one node at each time slot. Hence, the storage complexity of the proposed algorithm is linear, i.e., $O(T)$. In addition, the attributes of all the nodes in the tree should be updated at each time slot, and hence the running time at each time slot is also linear. As the algorithm runs $T$ time slots, the computational complexity of the HBA algorithm is a quadratic complexity $O(T^2)$. With the terminating condition, the tree is a finite tree and hence both storage complexity and computational complexity are bounded.
\section{Regret Performance Analysis}\label{sec: regret_performance_analysis}
In this section, we analyze the upper bound on the cumulative regret for the proposed algorithm. For the tractability of regret analysis, we have the following two assumptions.
\begin{assumption}
		\textbf{(Weak Lipschitz)} For any $x$ around the optimal $x^\star$, there exist constants $c_H>0$ and $\alpha>0$ such that 
	\begin{equation}
	f^\star-f(x)\leq c_H\|x^\star-x\|^\alpha
	\end{equation}
\end{assumption}	
{where $f^\star=f(x^\star)$ represents the optimum of function $f(\cdot)$. }This assumption indicates that the reward function satisfies the week Lipschitz condition, which can avoid sharp valleys around the optimal point that induces high regret. Furthermore, the weak Lipschitz condition is mild, which only has the impact on the region in the vicinity of the optimal value. This assumption is well justified in many practical applications \cite{shen2018generalized}.

\begin{assumption}
\end{assumption}
\begin{enumerate}	
	\item 	{\textbf{(Bounded diameter)} For a region, $C_{h,j}$, of depth $h$, the diameter of the region is defined as $D(C_{h,j})=\underset{x,y\in C_{h,j}}{\max} q(x,y)$. The diameter of the region is upper bounded by $ \rho_1 \gamma^h$ for constants $\rho_1>0$ and $0<\gamma<1$. }
	
	\item 	\textbf{(Well-shaped region)} For a region, $C_{h,j}$, of depth $h$, the region contains a ball with a radius of $\rho_2 \gamma^h$ which locates in the center of $C_{h,j}$.
	\end{enumerate}

The bounded diameter condition is to upper bound the maximum variation of $f(x)$ within the region $C_{h,j}$. In contrast, the well-shaped region condition is to lower bound the minimum variation of $f(x)$ within the region $C_{h,j}$. Note that any region in the reward function satisfies the bounded diameter and well-shaped region conditions \cite{NIPS}, which are utilized to bound the cumulative regret in the following analysis. 

%

%

\begin{definition}
	\textbf{$\epsilon$-optimal}. Let $f_{h,j}^\star= \underset{x \in {C_{h,j}}}{\max} f(x)$ be the optimal reward in $C_{h,j}$. If $f_{h,j}^\star>f^\star-\epsilon_{h,j}$, $C_{h,j}$ is the $\epsilon_{h,j}$-optimal region. 
\end{definition}

For example, if $\epsilon_{h,j}=0$, $C_{h,j}$ is the optimal region where optimal value $x^\star$ locates. Otherwise, if $\epsilon_{h,j}>0$, $C_{h,j}$ is a sub-optimal region. Let $\epsilon_{h,j}$ represent the \emph{suboptimality} of $(h,j)$. 

To obtain the regret bound, we first provide the following lemma.
\begin{lemma}\label{lemma:expected_N_bound}
	For any node $(h,j)$ whose suboptimality is larger than $\rho_1 \gamma^h$, the expected number of times that $(h,j)$ has been visited until time slot $T$, is upper bounded by
	\begin{equation}
	\mathbb{E}\left[N_{h,j}(T)\right] \leq\frac{8\sigma^2 \log T}{\left(\epsilon_{h,j} -\rho_1 \gamma^h\right)^2}+c
	\end{equation}
	where $c$ is a constant.
\end{lemma}
\begin{proof}
The detailed proof is given in Appendix \ref{appendix:lemma bound proof}.
\end{proof}
\begin{remark}
		From Lemma \ref{lemma:expected_N_bound}, the number of times that a suboptimal node has been visited logarithmically increases with time, which implies the cumulative regret of the proposed algorithm is sublinear. In addition, the number of times that a suboptimal node has been visited, depends on the variance of the channel fluctuation. A larger variance of the channel fluctuation implies a more noisy wireless environment, which yields more exploration efforts to remove the reward uncertainty.
\end{remark}

 Based on above lemma, an upper bound is obtained in the following.
\begin{theorem}\label{theorem: regret_bound}
	The upper bound on the cumulative regret of HBA is 
	\begin{equation}
	R^\pi\left({T}\right)=O\left(\sqrt{T\log T}\right).
	\end{equation}
\end{theorem}
\begin{proof}
	The detailed proof is given in Appendix \ref{appendix:regret bound proof}.
\end{proof}
\begin{remark}
	Theorem \ref{theorem: regret_bound} indicates the expected cumulative regret of HBA is sublinear in the time horizon $T$, i.e., $\lim\limits_{T \to \infty}R^\pi(T)/T=0$. Since the per-slot regret diminishes over time, the proposed algorithm is asymptotically optimal. Hence, the proposed algorithm converges to the optimal beam over time. Moreover, for finite time horizon $T$, the regret bound characterizes the convergence speed of the proposed algorithm. 


\end{remark}

\section{Simulation Results}\label{sec: simulation results}
\subsection{Simulation Setup}

We simulate an IEEE 802.11ad system, operating at 60 GHz with a bandwidth of $2.16$ GHz \cite{wu2019beef}. Consider an outdoor scenario, such as university campus, where the transmission distance between the transmitter and the receiver is set to 20 m unless otherwise specified. The average effective isotropically radiated power (EIRP) $P_{e}$ is fixed at 50 dBm\footnote{For outdoor applications with the high antenna gain, the average EIRP limit is up to 82 dBm \cite{FCC}.}, which is consistent with FCC regulations for 60 GHz unlicensed bands \cite{FCC, du2017much}. Taking the directional antenna gain into consideration, the transmit power is $P=P_{e}-10\log_{10}N$. For instance, the transmit powers are set to around 32 dBm and 23 dBm for 64 and 512 antenna arrays, respectively. It is worth noting that the mmwave channel is sparse, and hence we set the maximum number of channel paths to 5, which consists of one dominant LOS path and four NLOS paths. 
For the LOS path, the path loss is modeled as 
\begin{equation}
PL(dB)=32.5+20\log_{10}(f)+10\xi \log_{10}(d)+\chi
\end{equation}
where $f$, $\xi$, $d$, and $\chi$ represent the carrier frequency, path loss exponent, transmission distance, and shadow fading, respectively. The shadow fading follows $N(0,\sigma^2)$ where $\sigma$ is set to 2 dB \cite{Channel_model}. Note that the channel fluctuation in the simulation is mainly caused by the shadow fading. In addition, according to practical in-field measurements, NLOS paths suffer around 10 dB more path loss than the LOS path \cite{maltsev2009experimental}. We assume that the extra NLOS path loss follows a uniform distribution within $[7,13]$ dB. 
 Furthermore, for the HBA algorithm, the RSS within $[-80,-20]$ dBm is mapped to a reward within $[0,1]$. The algorithm parameters, $\rho_1$, $\gamma$, and $\zeta$ are set to 3, 0.5, and 0.1, respectively, based on extensive simulation trials. Important simulation parameters are listed in Table \ref{Simulation parameters}. We evaluate the performance via Monte-Carlo simulations. Simulation results are averaged based on 50000 samples with different channel fading and locations. The proposed HBA algorithm is compared to the following benchmarks: 

\begin{table}[t]
	\small
	\centering
	\caption{Simulation parameters.}
	\label{Simulation parameters}
	\begin{tabular}{cc|cc}
		\hline
		\hline
		\textbf{Parameter} & \textbf{Value}&\textbf{Parameter} & \textbf{Value} \\
		\hline
		Noise spectrum density $(N_o)$& $ -174$ dBm/Hz & System bandwidth $(W)$& $2.16$ GHz\\
		Carrier frequency $(f)$ & $ 60$ GHz &Path loss exponent $(\xi)$& 1.74\\ 
		Shadowing fading variance $(\sigma)$ & 2 dB&Signal range & $\left [-80, -20 \right ]$ dBm\\
		SSW frame duration $(T_{SSW})$ & 15.8 \emph{us}&Beacon interval duration $(T_{BI})$& 100 \emph{ms}\\
		Number of beams $(N)$& \{8-512\} & EIRP $(P_{e})$ & 50 dBm\\
		Number of paths $(L)$ & \{1-5\}&Algorithm parameters $(\rho_1,\gamma)$ & $(3,0.5)$\\
		Terminating condition threshold $(\zeta)$ & 0.1&Time horizon $(T)$ & 1000 time slots\\
		Extra NLOS path loss & $U(7,13)$ dB&Transmission distance $(d)$ & $20$ m\\
		\hline
		\hline
	\end{tabular}
\end{table}

%


%

\begin{itemize}

	\item \textbf{IEEE 802.11ad} \cite{Ad_standard}: In this industrial method, one side (transmitter or receiver) scans the beam space, while the other side keeps omni-directional.
	\item \textbf{UCB} \cite{auer2002finite}: The celebrated algorithm selects the beam without exploiting both correlation structure and prior knowledge. {The confidence margin is $\eta_u\sqrt{{2\log t}/{N_{b_i}(t)}}$, where the learning rate $\eta_u$ is set to 0.2 based on extensive simulation trials.}
	\item \textbf{Unimodal beam alignment (UBA)} \cite{infocom2018efficient}:  The algorithm exploits the unimodal structure among beams to perform BA. 
	Hence, it works in a ``hill-climbing" manner, which selects the best beam among the neighboring beams at each time slot. 
	
	\item \textbf{HOO} \cite{NIPS}: The algorithm selects the beam by exploiting beam correlation, without the prior knowledge. The confidence margin is $\eta_h\sqrt{2\log t/N_{h,j}{(t)}}+c_1\gamma^h$. {Here, the learning rate $\eta_h$ is set to 0.1, which is chosen based on extensive simulation trials.}
	%
\end{itemize}

\subsection{Regret Performance}
\begin{figure*}[t]
	\centering
	\renewcommand{\figurename}{Fig.}
	\begin{subfigure}[Cumulative regret performance comparison.]{
			\label{fig:cumulative_regret_L=2}
			\includegraphics[width=0.45\textwidth]{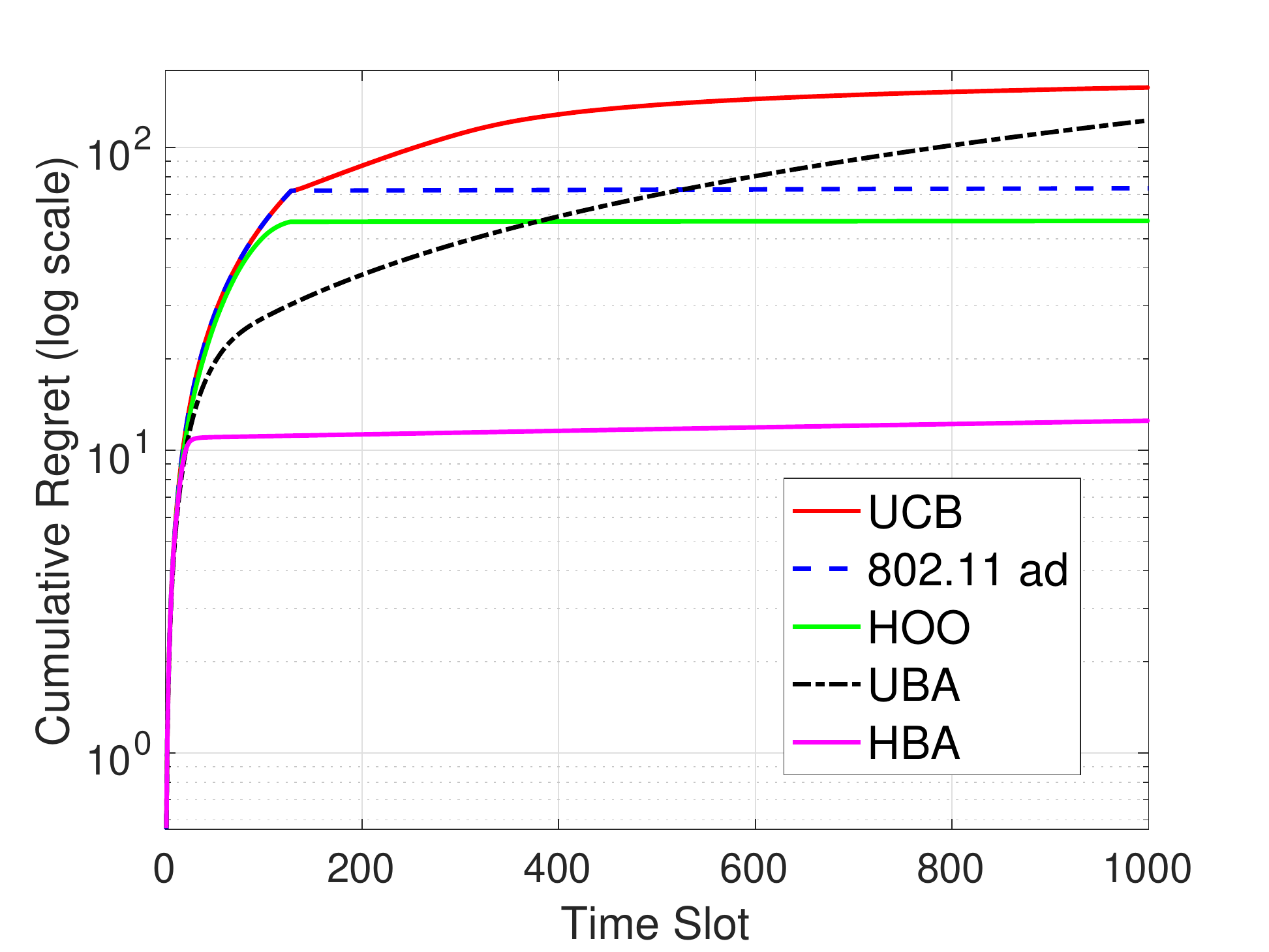}}
	\end{subfigure}%
	~
	\begin{subfigure}[Impact of the channel fluctuation distribution and variance.]{
			\label{fig:HBA_regret_vs_distribution_and_variance}
			\includegraphics[width=0.45\textwidth]{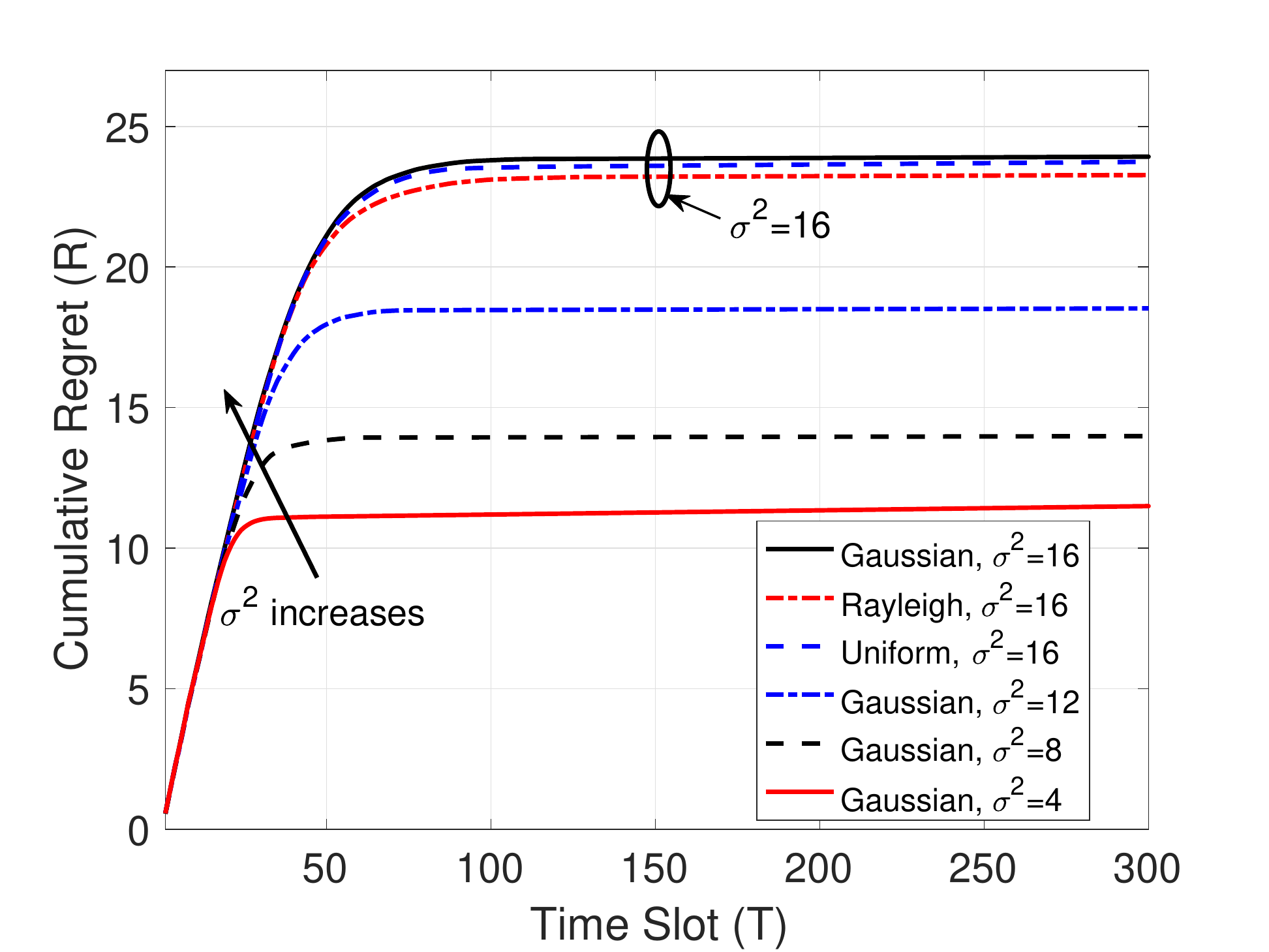}}
	\end{subfigure}
	\caption{Cumulative regret performance in the multipath channel.}
	\label{Fig:Cumulative regret performance with parameters}
	\vspace{-0.7cm}
\end{figure*}
Figure \ref{fig:cumulative_regret_L=2}  shows the cumulative regret performance in two-path channels. 
Several important observations can be obtained from simulation results. First of all, HBA significantly outperforms other benchmarks. A ``bounded regret" behavior is observed, which complies with the theoretical results in Theorem \ref{theorem: regret_bound}. In addition, HBA converges much faster than other benchmarks. Specifically, HBA only takes around 25 time slots to converge to the optimal beam. This is because HBA exploits both correlation structure and prior information to accelerate the BA process, while other benchmarks only exploit correlation structure or not. It is interesting to note that, as time goes by, the UBA algorithm performs even worse than the BA method in IEEE 802.11ad which does not exploit the correlation structure. The reason is that the UBA algorithm is designed based on the unimodal structure among beams, while the reward function evolves to a multimodal structure in the multipath channel. This model mismatch results in worse performance than not exploiting the correlation structure at all.


We further evaluate the impact of the channel fluctuation distribution on the regret performance in Fig. \ref{fig:HBA_regret_vs_distribution_and_variance}. To evaluate the dependency of the Gaussian distribution, the performance under Gaussian distribution is compared to that under two well-adopted non-Gaussian distributions, i.e., uniform distribution and Rayleigh distribution. The performance under non-Gaussian settings is very close to that under the Gaussian distribution, which means that the proposed algorithm can be applied in various settings. Furthermore, the impact of the channel fluctuation variance ($\sigma^2$) is studied in Fig. \ref{fig:HBA_regret_vs_distribution_and_variance}. As expected, the cumulative regret increases as the variance increases, because more exploration efforts are required in highly fluctuated channels. 

\subsection{Measurement Complexity and Beam Detection Accuracy}


\begin{figure*}[t]
	\centering
	\renewcommand{\figurename}{Fig.}
	\begin{subfigure}[{Number of beam measurements in the single-path channel}]{
			\label{fig:measurement_vs_array_size}
			\includegraphics[width=0.305\textwidth]{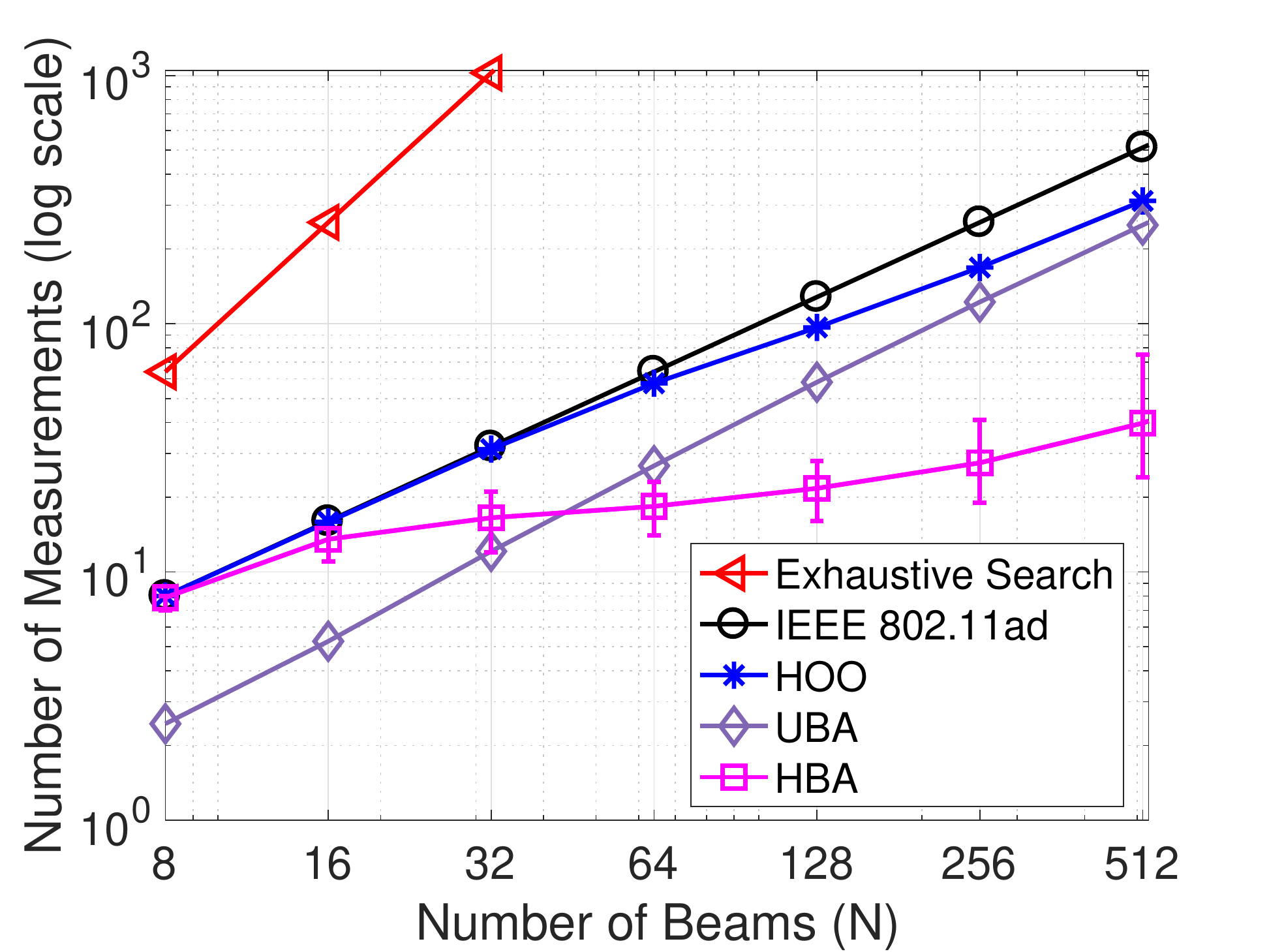}}
	\end{subfigure}%
	~
	\begin{subfigure}[Number of beam measurements in the multipath channel]{
			\label{fig:Num_measurement_vs_beams_paths.}
			\includegraphics[width=0.305\textwidth]{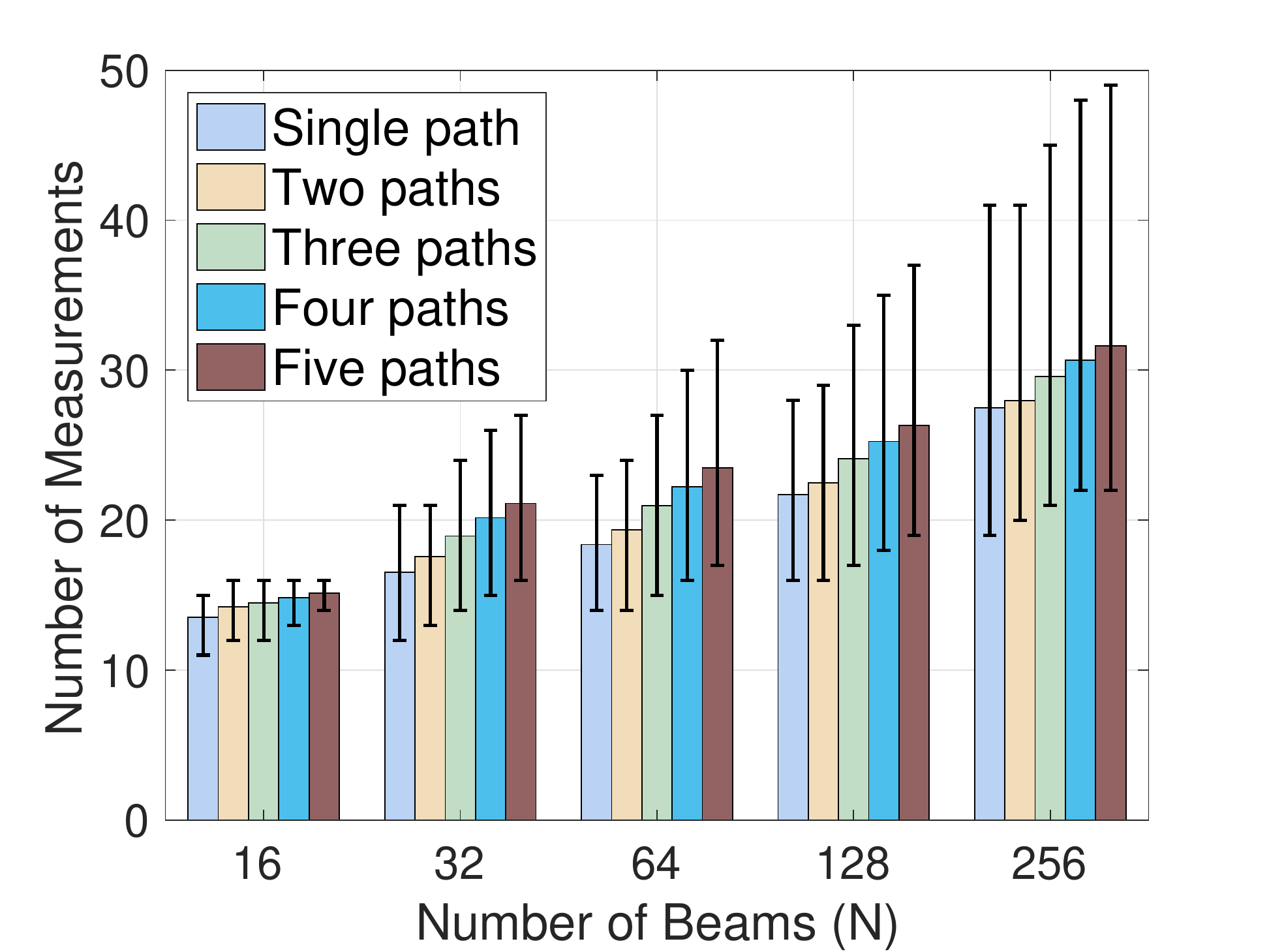}}
	\end{subfigure}%
	~
	\begin{subfigure}[Beam detection accuracy in the multipath channel]{
			\label{fig:HBA_accuracy}
			\includegraphics[width=0.305\textwidth]{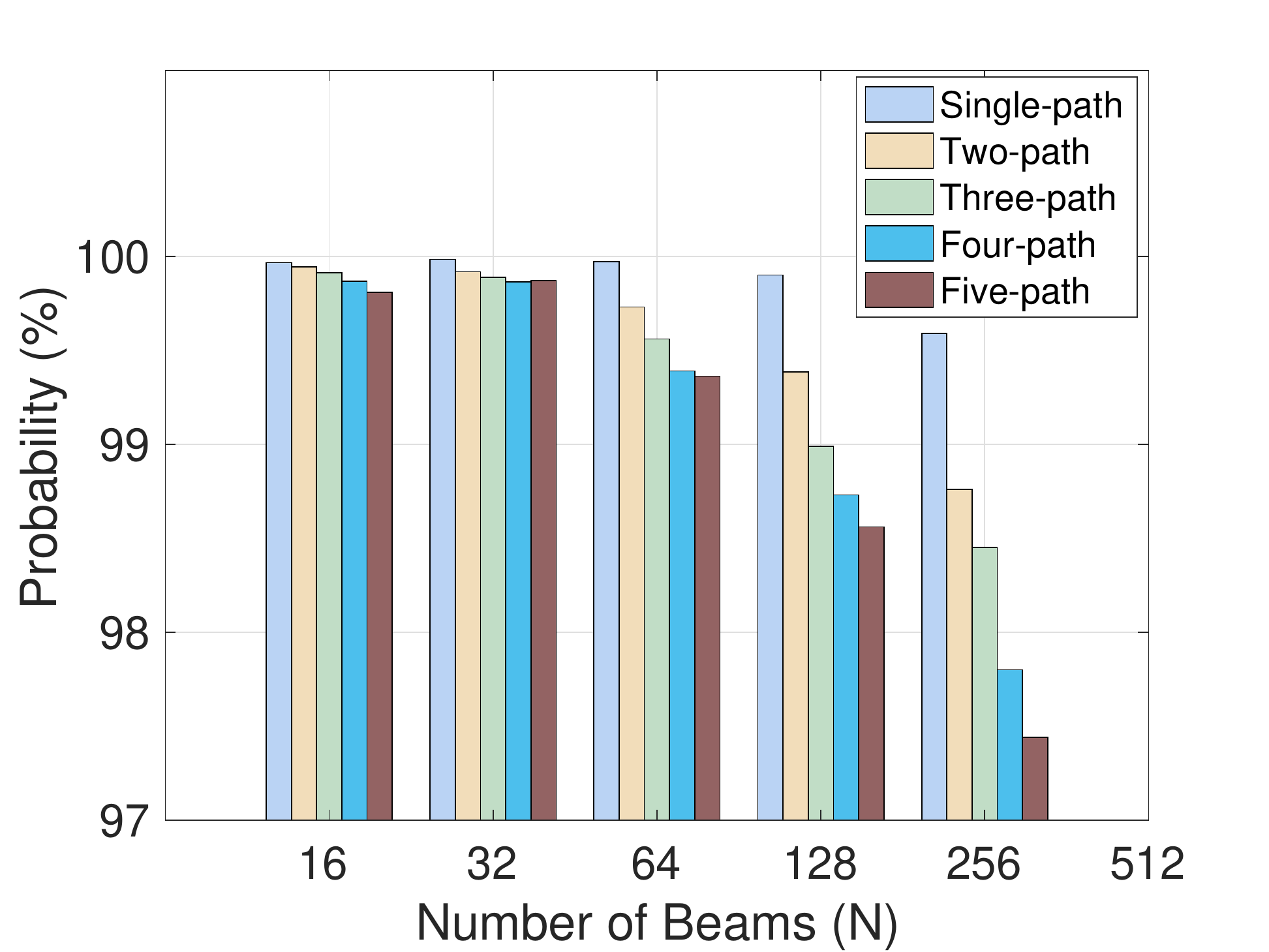}}
	\end{subfigure}
	\caption{Performance comparison with respect to the number of paths. Error bars show the 90 percentile performance.}
	\label{Fig:performance_paths_beams.}
\end{figure*}

The regret performance only reflects the bounded fact of regret, not necessarily the actual performance. Next, we evaluate the performance of HBA using following two metrics: the number of measurements and  beam detection accuracy.

{We first evaluate the scalability of the proposed algorithm with the number of beams in  single-path scenarios, as shown in Fig. \ref{fig:measurement_vs_array_size}. It is evident that the proposed algorithm significantly reduces the number of measurements as compared to the BA method in 802.11ad. }For a small number ($N=32$) of beams, the proposed algorithm reduces the number of measurements by 2 times as compared to the 802.11ad benchmark. Furthermore, the proposed algorithm achieves higher performance gains for larger numbers of beams. For instance, for a large number ($N=512$) of beams, the proposed algorithm only needs around 40 measurements to identify the optimal beam, which reduces the number of measurements by 12 times as compared to the 802.11ad benchmark. The reason is that, different from the BA method in 802.11ad that explores all the beams, the proposed algorithm only needs to explore a few beams by leveraging the correlation structure and the prior knowledge.  The results validate that the proposed algorithm is a scalable solution even with a large number of beams. { In addition, we compare the HBA algorithm with the UBA algorithm. It can be seen that the UBA algorithm performs better than the HBA algorithm when the number of beams is small ($N\leq 32$). However, when the number of beams is large, HBA performs much better than UBA. Since  UBA works in a ``hill-climbing" manner to find the optimal beam, the number of measurements required by UBA increases with the number of beams due to a longer path to the optimal point. To avoid exceedingly high BA latency, the BA performance for a large number of beams is crucial.  Thus, the proposed algorithm is more effective than the UBA algorithm when the number of beams is large. Besides, UBA does not work well in multipath scenarios, while the proposed algorithm does.}

\begin{figure*}[t]
	\centering
	\renewcommand{\figurename}{Fig.}
	\begin{subfigure}[Number of measurements]{
			\label{fig:measurement_vs_distance}
			\includegraphics[width=0.45\textwidth]{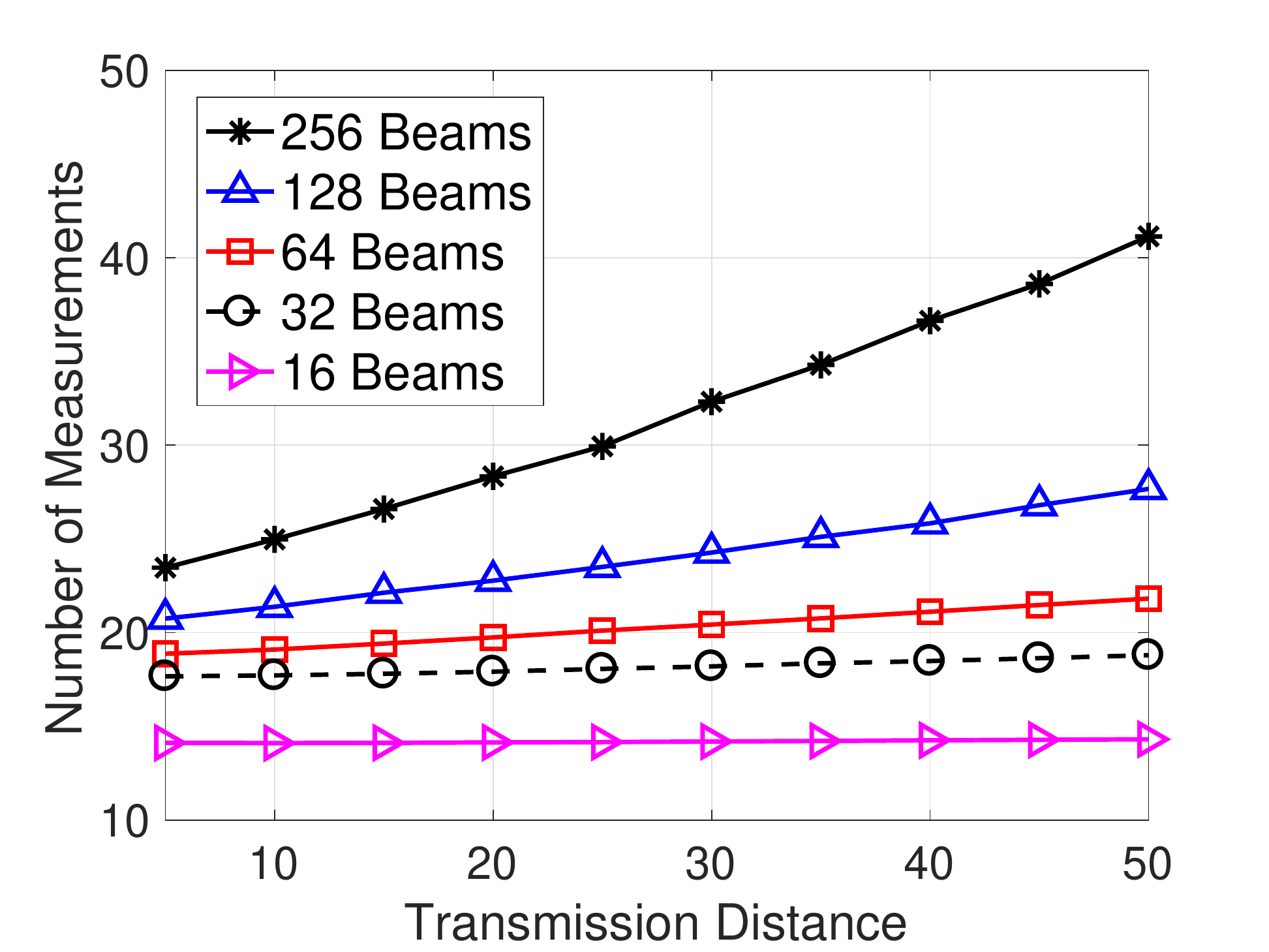}}
	\end{subfigure}
	~
	\begin{subfigure}[Beam detection accuracy]{
			\label{fig:accuracy_vs_distance}
			\includegraphics[width=0.45\textwidth]{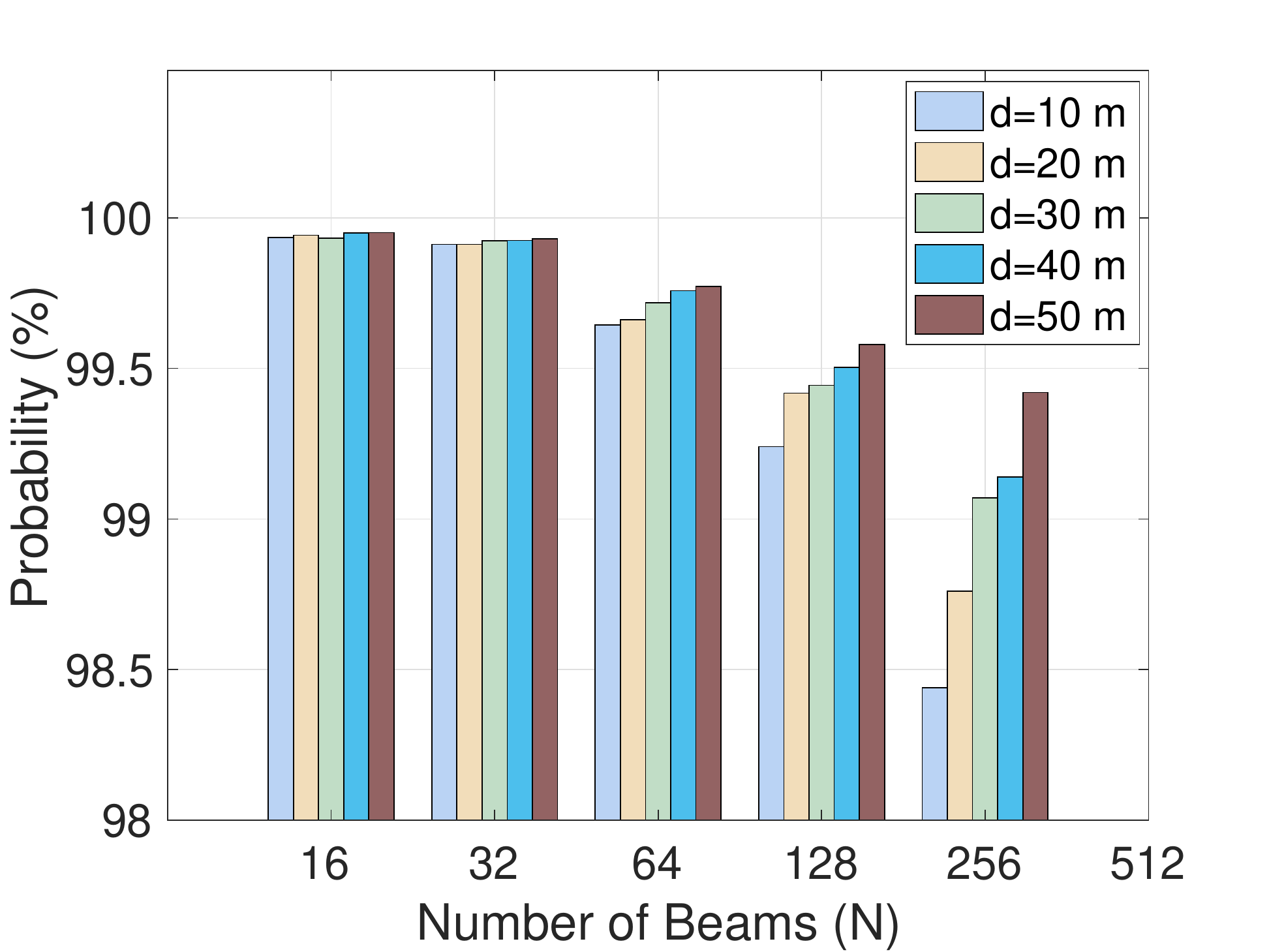}}
	\end{subfigure}
	\caption{Performance comparison with respect to transmission distance in two-path channels.}
	\label{Fig:Performance_vs_distance}
	\vspace{-0.7cm}
\end{figure*}

 As shown in Fig. \ref{Fig:performance_paths_beams.}, we further study the performance in multipath channels.  Due the inherent sparse characteristics of the mmwave channel, the number of paths is selected from 1 to 5. Firstly, the numbers of measurements in terms of the number of paths are compared in Fig. \ref{fig:Num_measurement_vs_beams_paths.}. It can be seen that the number of measurements increases slightly as the number of paths increases. For example, for a 128-beam case, the number of measurements in the five-path channel increases by 15\% as compared to that in the single-path channel. Secondly, beam detection accuracy performance is presented in Fig. \ref{fig:HBA_accuracy}. The HBA algorithm detects the optimal beam with a high probability, even in sophisticated multipath channels. Simulation results show that the beam detection accuracy is higher than 97\%, even in the worst case. In addition, the beam detection accuracy slightly decreases as the number of paths increases. For a large number ($N=256$) of beams, the beam detection accuracy decreases from 99.6\% in the single-path channel to 97.4\% in the five-path channel due to the sophisticated multipath channel.


Figure \ref{Fig:Performance_vs_distance} shows the impact of the transmission distance on the performance. 
We first observe that the number of measurements increases in terms of the transmission distance, as shown in Fig. \ref{fig:measurement_vs_distance}. Specifically, the number of measurements increases by 32\% as distance increases from 5 meters to 50 meters  for $N=128$. Because the RSS is weaker for a longer distance such that limited information can be extracted from nearby beams. Hence, the proposed algorithm needs to explore more beams to identify the optimal beam for remote users. Even for remote users, the proposed BA algorithm performs better than the 802.11ad benchmark. When the distance increases to 50 meters, our algorithm needs about 44 measurements for $N= 256$, which still reduces the number of measurements by 5.8 times as compared to the 802.11ad benchmark. Finally, the beam detection accuracy is presented in Fig. \ref{fig:accuracy_vs_distance}. Even in the low SNR case, the proposed algorithm can detect the optimal beam with a high probability.




\begin{figure*}[t]
	\centering
	\renewcommand{\figurename}{Fig.}
	\begin{subfigure}[Number of beam measurements]{
			\label{fig:inaccurate_prior_measurements.}
			\includegraphics[width=0.45\textwidth]{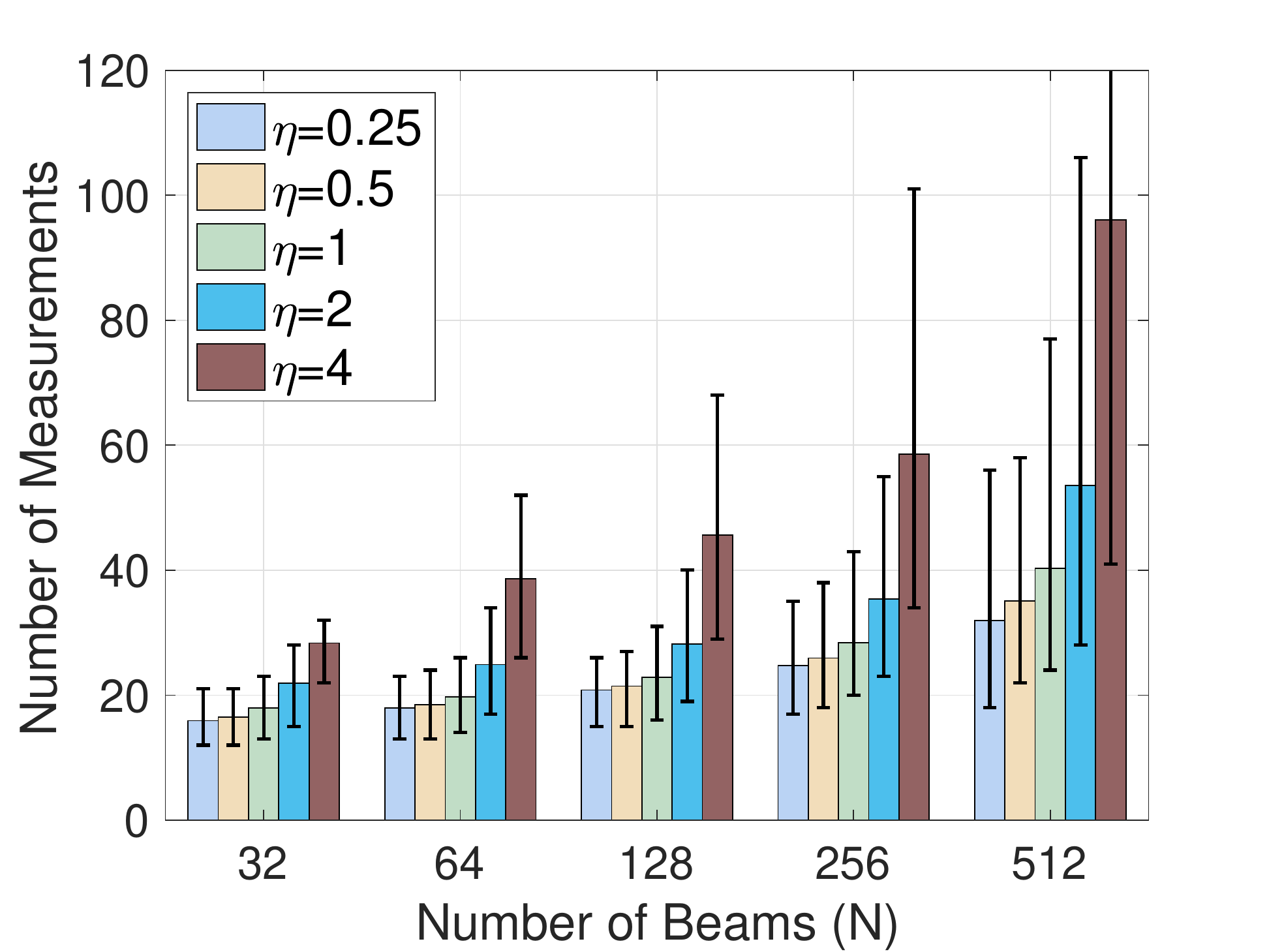}}
	\end{subfigure}%
	~
	\begin{subfigure}[Beam detection accuracy]{
			\label{fig:inaccurate_prior_probability}
			\includegraphics[width=0.45\textwidth]{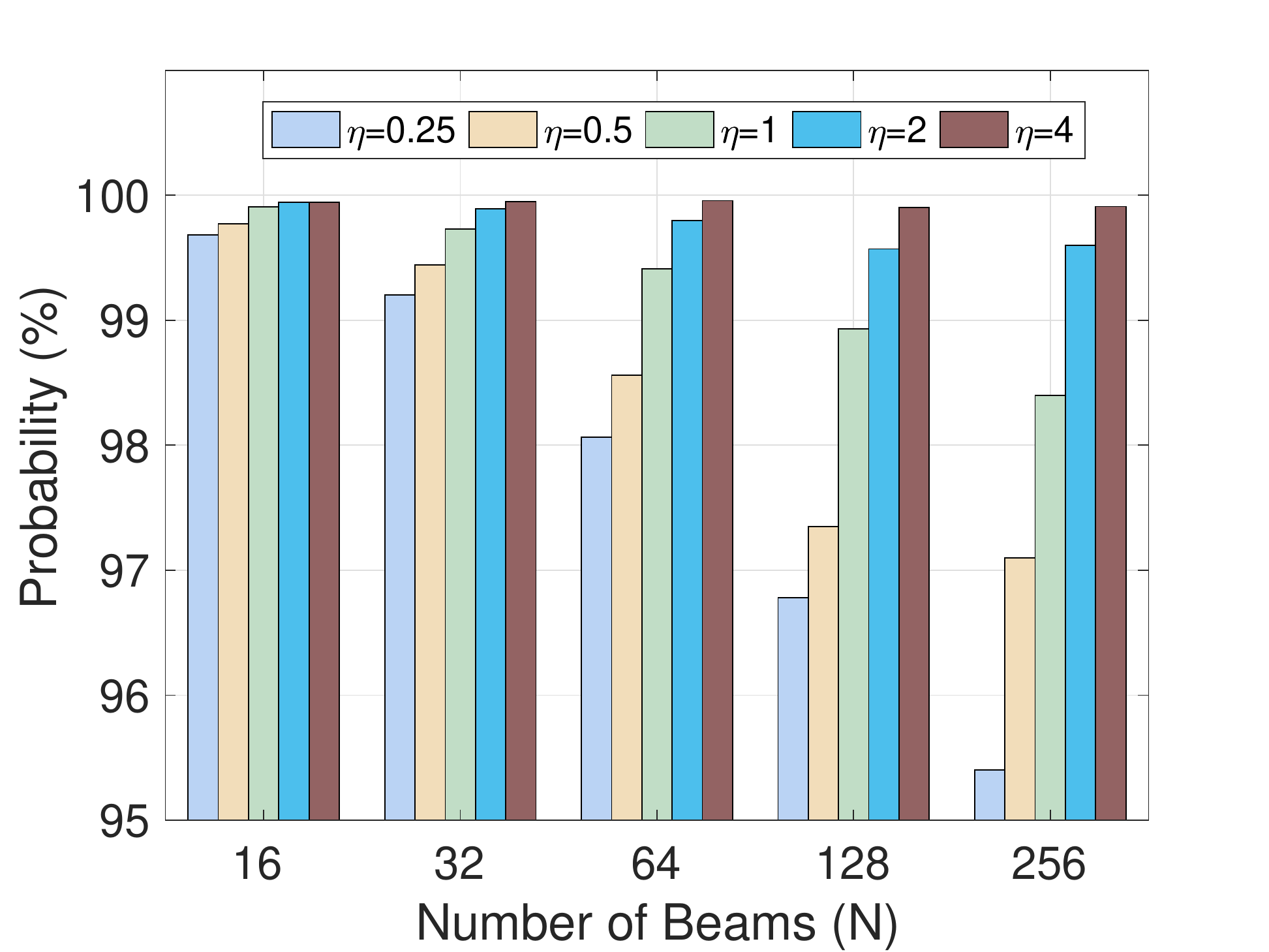}}
	\end{subfigure}
	\caption{Performance comparison with coarse prior knowledge in two-path channels.}
	\label{Fig:inacurate_prior.}
	\vspace{-0.5cm}
\end{figure*}

For implementation consideration, Fig. \ref{Fig:inacurate_prior.} presents the performance of HBA under coarse prior knowledge conditions. The metric of the coarse prior knowledge is defined as a ratio between the estimated variance $(\sigma^2_e)$ and the accurate one, i.e., $\eta={\sigma^2_e}/{\sigma^2}$. Hence, the coarse prior knowledge can be divided into two categories: the underestimated prior knowledge when $\eta<1$ and the overestimated prior knowledge when $\eta>1$. We can see from Fig. \ref{fig:inaccurate_prior_measurements.} that the number of  measurements increases as $\eta$ increases from 0.25 to 4. Specifically, for a 256-beam case, the HBA algorithm with the overestimated prior knowledge for $\eta=4$ requires more beam measurements as compared to that with accurate prior knowledge. Overestimating prior knowledge results in a larger confidence margin to accommodate reward uncertainty, such that more exploration efforts are needed and better beam detection accuracy can be achieved, as shown in Fig. \ref{fig:inaccurate_prior_probability}. In contrast, when prior knowledge is underestimated, the number of measurements is slightly smaller than that with accurate prior knowledge, while the beam detection accuracy decreases due to insufficient exploration efforts. More importantly, even with the coarse prior knowledge, the proposed algorithm can substantially reduce the number of measurements as compared to benchmarks, and achieve high beam detection accuracy. For a 256-beam case, even in the worst case, the proposed algorithm reduces the number of measurements by 6 times in comparison with the BA method in 802.11ad.

\subsection{BA Latency}


{Practical BA latency needs to take the 802.11ad protocol into consideration, which is different from a simple product of the number of measurements and the duration of each measurement. In the protocol, BA must be performed in the associated beamforming training (A-BFT) stage, which contains 8 A-BFT slots, and each A-BFT slot contains 16 sector sweep (SSW) frames. Each SSW frame can only provide one measurement for one beam and has a duration about 15.8 \emph{us} \cite{Ad_standard, ay_doc_short_SSW}. If the BA process cannot be finished in the A-BFT stage of the current beacon interval (BI), this BA process has to wait for the A-BFT stage in the next BI, which increases the BA latency for a whole BI duration. In the simulation, the duration of BI is set to 100 \emph{ms} \cite{Ad_standard}.  In addition, since the HBA algorithm requires the feedback of RSS of the selected beam at each round, the feedback latency should also be incorporated into the calculation of BA latency. The duration of a feedback frame at each round is about 1 \emph{us} in 802.11ad \cite{sur2018towards}. Taking the above protocol and the feedback latency into consideration, BA latency is calculated based on the average number of measurements. Table \ref{T:beam_alignment_latency} presents the BA latency with different numbers of beams in the two-path channel. As expected, the BA latency increases as the number of beams increases. For the case with one user, the proposed algorithm reduces the BA latency significantly as compared to the BA method in 802.11ad. In particular, for a large number ($N=256$) of beams, the BA latency drops from 106.07 \emph{ms} to only 0.94 \emph{ms}. This is because the BA process with the proposed algorithm can be finished in one BI as a small number of measurements is required to identify the optimal beam. Furthermore, a larger performance gain can be observed in the four-user case. In contrast to the BA method in 802.11ad which incurs more than 700 \emph{ms} latency for a 256-beam phase arrays, the proposed algorithm takes about 2.35 \emph{ms}, which corresponds to two orders of magnitude gain.} 

\begin{table}[t]
	\small
	\centering
	\caption{{BA latency with different numbers of beams in multipath channels.}}
	\label{T:beam_alignment_latency}
	\vspace{-0.3cm}
	\begin{center}
		\begin{tabular}{| c | c | c | c | c |}
			\hline
			& \multicolumn{2}{ c |}{{\textbf{One user}}}  & \multicolumn{2}{ c |}{{\textbf{Four-user}}} \\ 
			\cline{2-5}
			{\textbf{$N$}}	& {\textbf{802.11ad}} & {\textbf{HBA}} & {\textbf{802.11ad}} & {\textbf{HBA}} \\
			\hline			
			{16} &  {0.51 ms}&  { 0.48 ms}&   {1.26 ms}& {1.19 ms}\\ \hline
			{32} &   {1.01 ms}&  {0.59 ms} & {2.53 ms}& {1.47 ms}\\ \hline
			{64}&  {2.02 ms}&  {0.65 ms}&  {103.03 ms}& {1.63 ms}\\ \hline
			{128} &  {4.04 ms}&   {0.76 ms}&  {304.04 ms}& {1.89 ms}\\ \hline
			{256} &  {106.07 ms}&  {0.94 ms}& { 706.07 ms}& {2.35 ms}\\ \hline
		\end{tabular}
	\end{center}
\end{table}

%
%


\section{Conclusion}\label{sec: conclusion}
In this paper, we have investigated the BA problem in mmwave systems to find the optimal beam pair. We have developed HBA, a learning algorithm which leverages the inherent correlation structure among beams and the prior knowledge on the channel fluctuation to accelerate the BA process. The proposed HBA algorithm can identify the optimal beam with a high probability using a small number of beam measurements, even when the number of beams is large. HBA can be applied to meet the demand of delay-sensitive Gbps applications, such as cordless virtual reality gaming. Beyond the BA problem, the design principle of leveraging correlation structure is useful in other  optimization problems in wireless networks, such as power allocation and interference mitigation. 
For our future works, it would be interesting to extend the proposed algorithm to mobile scenarios, where the environment is highly dynamic and delay requirement is more stringent. In such scenario, the main challenge lies in extracting information from the real-time environment to speed up BA.


\appendix
\subsection{Proof of Theorem \ref{lemma:unimodality_reward_function}}\label{appendix:theorem multimodality}


	
	According to \eqref{equ:antenna_directivity_gain}, the maximum RSS can be achieved with the minimum angular misalignment denoted by, $\delta=\omega_{i^\star}-\vartheta$, where $\omega_{i^\star}$ is the spatial angle for the optimal transmit beam. 
	Hence, 
	$D\left(\omega_i-\vartheta_{}\right)$ can be rewritten as
	\begin{equation}\label{equ:D_monotonical}
	\begin{split}
	D\left(\omega_i-\vartheta_{}\right)
	&=D\left(\delta +\frac{2(i-i^\star)}{N}\right)
	=\frac{\sin^2({N  \pi d \delta }/{\lambda})}{ \sin^2\left({\pi d \left(\delta +\frac{2(i-i^\star)}{N}\right)}/{\lambda}\right)}, \forall b_i\in \mathcal{B}.
	\end{split}
	\end{equation}
	
	From simple analysis in \eqref{equ:D_monotonical}, $D\left(\omega_i-\vartheta_{}\right)$ monotonically increases in $[i^\circ,i^\star]$ and decreases in $[i^\star, i^\star+\frac{N}{2}]$, where $i^\circ=i^\star-\frac{N}{2}$. 
	Hence, the mean RSS function over the beam space increases along path $(b_{i^\circ},b_{i^\circ+1},..., b_{i^\star})$ and decreases along path $(b_{i^\star},b_{i^\star+1},..., b_{i^\circ-1})$, i.e., $r(b_{i^\circ})<r(b_{i^\circ+1})<...<r(b_{i^\star})>...>r(b_{i^\circ-2})>r(b_{i^\circ-1})$.
	With the definition of the unimodality structure, the mean RSS function is unimodal over the beam space in the single-path channel, and the theorem statement follows. 

\subsection{Proof of Corollary \ref{proposition:multimodal}}\label{appendix:tmultimodality}

	Similar to \eqref{equ:antenna_directivity_gain}, the mean RSS in the multipath channel is represented by
	\begin{equation}
	\begin{split}
	\mathbb{E}\left[	r({b_i})\right]
	&=\underbrace{\frac{Pg_0^2 }{N }D\left(\omega_i-\vartheta_{0}\right)}_{\text{LOS component}}+\underbrace{\sum_{l=1}^{L-1}\frac{Pg_l^2 }{N }D\left(\omega_i-\vartheta_{l}\right)}_{\text{NLOS component}}+N_oW
	\end{split}
	\end{equation}
	
	Above equation indicates that the aggregated RSS consists of a LOS component and several NLOS components. For each individual path of the mmwave channel, the corresponding RSS function is unimodal function based on Theorem \ref{lemma:unimodality_reward_function}. Hence, the RSS function in the multipath channel is the aggregation of several unimodal functions, which can be considered as a multimodal function. Specifically, $L$ paths exist in the mmwave channel, which corresponds to $L$ peaks in the multimodal function. As the channel gain of the LOS path is  significantly larger than that of NLOS paths, i.e., $g_0^2> g_l^2$. Hence, the dominant peak corresponds to the LOS path while other peaks correspond to NLOS paths. Hence, the Corollary \ref{proposition:multimodal} is proved.

\subsection{Proof of Lemma \ref{lemma:expected_N_bound}}\label{appendix:lemma bound proof}

	For any integer $m>0$, according to the definition, the average times that node $(h,j)$ has been visited up to time slot $T$, is given by
	\begin{equation}\label{equ:E[H(T)]_bound}
	\begin{split}
&\mathbb{E}\left[	N_{h,j}(T) \right]
	=\mathbb{E}\left[\sum_{t=1}^{T}\mathbbm{1}_{(H_t,J_t)\in {C}_{h,j} }\right]\\
	&=\mathbb{E}\left[ \sum_{t=1}^{T}  \mathbbm{1}_{\{(H_t,J_t)\in {C}_{h,j} , N_{h,j}(t)\leq m \}} \right]
	+\mathbb{E}\left[ \sum_{t=1}^{T} \mathbbm{1}_{\{(H_t,J_t)\in {C}_{h,j} ,N_{h,j}(t)> m \}}\right]\\
	&\leq m+\mathbb{E}\left[\sum_{t=m+1}^{T}\mathbbm{1}_{\{(H_t,J_t)\in {C}_{h,j} , N_{h,j}(t)> m \}}\right]
	= m+\sum_{t=m+1}^{T} \mathbb{P}\left( \left(H_t,J_t\right)\in {C}_{h,j}, N_{h,j}(t)> m \right).
	\end{split}
	\end{equation}
	where $\mathbbm{1}_{\{\cdot\}}$ is the indicator function and $(H_t,J_t)\in {C}_{h,j}$ denotes the selected node $(H_t,J_t)$ locates within ${C}_{h,j}$. The first equality is because $N_{h,j}(t)> m$ only occurs when $t$ is larger than $m$. 
	
	We apply a case study to obtain an upper bound of $\mathbb{E}\left[	N_{h,j}(T) \right]$. Assume node $(h,j)$ is selected at time slot $t$. The path from root node $(0,1)$ to $(h,j)$ is given by, $\mathcal{P}=\{ (0,1), (1, j_{1}^\star),..., (k, j_k^\star),(k+1, j_{k+1}^o),..., (h,j) \}$, where $k$ denotes the largest depth of the optimal node in the path. Before node $(k, j_{k}^\star)$, the optimal nodes are selected. For notation simplicity, we omit the time slot $t$ in $Q_{k,j}(t)$. 
	After traversing node $(k, j_{k}^\star)$, a sub-optimal node $(k+1, j_{k+1}^o)$ is selected instead of the optimal node $(k+1, j_{k+1}^\star)$ because the suboptimal node has a larger $Q$-value than the optimal node, i.e., $Q_{k+1, j^o}\geq Q_{k+1, j^\star }$. As $Q$-values increase along path $\mathcal{P}$, we have $Q_{k+1, j^\star }\leq Q_{k+1, j_{k+1}^o}\leq,..., \leq Q_{h, j}$. Note that $Q$-values are upper bounded by $E$-values according to the definition, such that 
	$Q_{k+1, j^\star }\leq E_{h, j}$. Further, event $ Q_{k+1, j^\star }\leq E_{h, j}$ can be interpreted as the union of two events, $\{Q_{k+1, j^\star }\leq f^\star\}\cup \{ E_{h, j}\geq f^\star\}$. Hence, the probability that $(H_t,J_t)$ locates within $C_{h,j}$ is upper bounded by
	\begin{equation}\label{equ: set1}
	\mathbb{P}\left((H_t,J_t)\in C_{h,j}\right) \leq \mathbb{P}\left(Q_{k+1, j^\star }\leq f^\star \right)+ \mathbb{P}\left(E_{h, j}\geq f^\star\right). 
	\end{equation}
	
	With the definition of $Q$-value, the $Q$-value of a node is the minimum value among the $E$-value of the node and $Q$-values of its child nodes. Hence, event $\{Q_{k+1, j^\star }\leq f^\star \}$ can be interpreted as the union of two new events, $\{E_{k+1, j^\star }\leq f^\star \} \cup \{Q_{k+2, j_{k+2}^\star }\leq f^\star \}$. Since event $\{Q_{k+2, j_{k+2}^\star }\leq f^\star \}$ can be further recursively expanded as $\bigcup\limits_{s=k+2}^{t-1}\{E_{s, j_s^\star }\leq f^\star \}$, we have
	\begin{equation}\label{equ: set2}
		\mathbb{P}\left(Q_{k+1, j^\star }\leq f^\star \right)\leq \sum \limits_{s=k+1}^{t-1} 	\mathbb{P}\left(E_{s, j_s^\star }\leq f^\star \right).
	\end{equation}
	
	Substituting \eqref{equ: set2} and  \eqref{equ: set1} into \eqref{equ:E[H(T)]_bound}, \eqref{equ:E[H(T)]_bound} can be rewritten as
	
	\begin{equation}\label{equ:N_h_j_bound }
	\begin{split}
	\mathbb{E}\left[	N_{h,j}(T) \right] 
	&\leq m+\sum_{t=m+1}^{T} \left( \sum_{s=k+1}^{t-1}\mathbb{P}\left( E_{s, j^\star }\left(t\right)\leq f^\star  \right) + \mathbb{P}\left( E_{h, j}\left(t\right)\geq f^\star,N_{h,j}(t)> m  \right)  \right).
	\end{split}
	\end{equation}
	
	The following analysis is to bound the three terms in \eqref{equ:N_h_j_bound } separately.

Firstly, since $m$ is an arbitrary integer, taking $m$ as the smallest integer that satisfies the condition $m\geq \frac{8\sigma^2 \log T}{\left(\epsilon_{h,j}-c_1\gamma^h \right)^2}$.  
	 Hence $m$ is bounded by
	\begin{equation}\label{equ:m_value}
	m \leq \frac{8\sigma^2 \log T}{\left(\epsilon_{h,j}-\rho_1\gamma^h \right)^2}+1.
	\end{equation}
	
Secondly, we aim to bound the first term $\mathbb{P}\left( E_{s, j^\star }\leq f^\star  \right)$. For the optimal nodes $(h, j^\star)$,  according to the definition of $E$-values, $E_{h, j^\star}=\infty$ when $N_{h,j^\star}=0$. Hence, event $E_{h, j^\star }\leq f^\star  $ only occurs when $N_{h,j}\geq1$. As a result, $\mathbb{P}\left( E_{h, j^\star }\leq f^\star  \right)$ can be rewritten as
\begin{equation}\label{equ:E_bound}
\begin{split}
&\mathbb{P}\left( E_{h, j^\star }\leq f^\star ,N_{h,j}\geq1 \right)=\mathbb{P}\left( R_{h, j^\star }+\sqrt{\frac{2\sigma^2 \log t}{N_{h,j^\star}}}+\rho_1\gamma^h \leq f^\star ,N_{h,j^\star}\geq1 \right)\\
&=\mathbb{P}\left(  \left( f^\star- R_{h, j^\star } -\rho_1\gamma^h\right)N_{h,j^\star}\geq  \sqrt{2\sigma^2N_{h,j^\star} \log t}  ,N_{h,j^\star}\geq1 \right)\\
&\stackrel{{(a)}}{=}\mathbb{P}\left( \sum_{s=1}^{t} \left(f^\star -f(X_s) +\rho_1\gamma^h \right)\mathbbm{1}_{(H_t,J_t)\in {C}_{h,j^\star}}  \right. 
\\&\left. + \sum_{s=1}^{t} \left( f(X_s)-Y_s\right) \mathbbm{1}_{(H_t,J_t)\in {C}_{h,j^\star}} \geq  \sqrt{2\sigma^2 N_{h,j^\star}\log t}  ,N_{h,j^\star}\geq1 \right)\\
&\stackrel{{(b)}}{\leq}\mathbb{P}\left( \sum_{s=1}^{t} \left(f(X_s)-Y_s\right) \mathbbm{1}_{(H_t,J_t)\in {C}_{h,j^\star}} \geq  \sqrt{2\sigma^2N_{h,j^\star} \log t},N_{h,j^\star}\geq1 \right)\\
&\stackrel{{(c)}}{=}\mathbb{P}\left( \sum_{p=1}^{N_{h,j^\star}}\left( \tilde{Y}_p-f(\tilde{X}_p) \right)  \geq  \sqrt{2\sigma^2 N_{h,j^\star}\log t}  ,N_{h,j^\star}\geq1 \right).
\end{split}
\end{equation}
In \eqref{equ:E_bound}, the first step follows from the definition of $E$-value in \eqref{equ:E_value}; $(a)$ is obtained from the definition of $ N_{h,j^\star}$, where $X_s, \forall s=1,2,...,t-1$ denotes the sequentially selected beams up to time $t-1$ and the corresponding reward sequence is represented by $Y_s$; $(b)$ follows from the fact that $ f^\star-f(X_t) -\rho_1\gamma^h<0$ holds for all the beams in the optimal region ${C}_{h,j^\star}$; $(c)$ is because the definition of a new beam selection sequence $\tilde{X}_p, \forall p=1,2,3,...$ whose corresponding reward sequence is $\tilde{Y}_p$.  

Let $T_p=\min\{t:N_{h,j}(t)=p\}$ represent the time sequence for the selected node in $C_{h,j}$. The sequentially selected beams can be represented by a new sequence $\tilde{X}_p=X_{T_p},\forall p=1,2,3,...$, and \eqref{equ:E_bound} can be further bounded by
\begin{equation}\label{equ:finalE_bound}
\begin{split}
&\mathbb{P}\left( \sum_{p=1}^{N_{h,j_h^\star}}\left( \tilde{Y}_p-f(\tilde{X}_p) \right)  \geq  \sqrt{2\sigma^2 N_{h,j^\star}\log t}  ,N_{h,j_h^\star}\geq1 \right)\\
&\stackrel{{(a)}}{\leq}\sum_{s=1}^{t}\mathbb{P}\left( \sum_{p=1}^{s}\left( \tilde{Y}_p-f(\tilde{X}_p) \right)  \geq  \sqrt{2\sigma^2 s\log t}   \right)
\stackrel{{(b)}}{\leq}\sum_{s=1}^{t} \exp \left( -\frac{4\sigma^2 s\log t}{  s\sigma^2} \right)=t^{-3}.
\end{split}
\end{equation}
In \eqref{equ:finalE_bound}, $(a)$ can be acquired via the union bound that takes all possible values of $N_{h,j_h^\star}$; as $\tilde{D}_p =\tilde{Y}_p- f(\tilde{X}_p)$ can be considered as martingale differences, $(b)$ is obtained via the Hoeffding-Azuma inequality \cite{NIPS}
\begin{equation}\label{equ: the Hoeffding-Azuma inequality }
\mathbb{P}\left(  \sum_{p=1}^{ k} \tilde{D}_p \geq t \right)\leq \exp\left(-\frac{2 t^2}{\sum_{p=1}^{k} \sigma^2 }\right) .
\end{equation} 

Thirdly, for suboptimal nodes $(h,j)$, the upper bound of $\mathbb{P}\left( E_{h, j}\geq f^\star, N_{h,j}> m \right)$ can be obtained via a similar method of bounding $\mathbb{P}\left( E_{h, j^\star }\leq f^\star ,N_{h,j}\geq1 \right)$, such that
	\begin{equation}\label{equ:bound_E>f^star}
	\begin{split}
	&\mathbb{P}\left( E_{h, j}\geq f^\star, N_{h,j}> m \right)
	=\mathbb{P}\left( R_{h, j}+\sqrt{\frac{2\sigma^2\log t}{N_{h,j} }}+\rho_1\gamma^h  \geq f_{h,j}^\star+ \epsilon_{h,j}, N_{h,j}> m \right)\\
	&\stackrel{{(a)}}{\leq} \mathbb{P}\left( R_{h, j}  \geq f_{h,j}^\star+\frac{\epsilon_{h,j}-\rho_1\gamma^h}{2},N_{h,j}> m  \right)\\
	&= \mathbb{P}\left(  \left( R_{h, j} - f_{h,j}^\star\right)N_{h,j}  \geq  \frac{\epsilon_{h,j}-\rho_1\gamma^h}{2}N_{h,j}, N_{h,j}> m  \right)\\
	&\stackrel{{}}{=}\mathbb{P}\left(  \sum_{s=1}^{t} \left(Y_s- f_{h,j}^\star\right) \mathbbm{1}_{(H_s,J_s)\in {C}_{h,j}} \geq N_{h,j} \frac{\epsilon_{h,j}-\rho_1\gamma^h}{2},   N_{h,j}> m \right)\\
	&\leq \mathbb{P}\left(  \sum_{s=1}^{t} \left(Y_s- f(X_s)\right) \mathbbm{1}_{(H_s,J_s)\in {C}_{h,j}} \geq N_{h,j} \frac{\epsilon_{h,j}-\rho_1\gamma^h}{2},  N_{h,j}> m \right)\\
	&\stackrel{{(b)}}{=} \mathbb{P}\left(   \sum_{p=1}^{ N_{h,j}} \left(\hat{Y}_p- f(\hat{X}_p)\right) \geq N_{h,j} \frac{\epsilon_{h,j}-\rho_1\gamma^h}{2}, N_{h,j}> m  \right)
	\end{split}
	\end{equation} 
	In \eqref{equ:bound_E>f^star}, $(a)$ due to the substitution of $ N_{h,j}(t)\geq \frac{8\sigma^2 \log t}{\left(\epsilon_{h,j}-\rho_1\gamma^h \right)^2}$ where $m \geq \frac{8\sigma^2 \log t}{\left(\epsilon_{h,j}-\rho_1\gamma^h \right)^2}$;  $(b)$ is obtained via a similar method as  \eqref{equ:E_bound}$(c)$, where a new beam sequence $\{\hat{X}_1,\hat{X}_2,...,\hat{X}_p \}$ is formed to represent the sequentially selected beams in ${C}_{h,j}$. Next, \eqref{equ:bound_E>f^star} can be further bounded by	
	\begin{equation}\label{equ: finalbound_E>f^star}
	\begin{split}
	&\mathbb{P}\left(   \sum_{p=1}^{ N_{h,j}} \left(\hat{Y}_p- f(\hat{X}_p)\right) \geq N_{h,j} \frac{\epsilon_{h,j}-\rho_1\gamma^h}{2}, N_{h,j}> m  \right)\\
	&\stackrel{{(a)}}{\leq} \sum_{k=m+1}^{t}\mathbb{P}\left(   \sum_{p=1}^{ k} \left(\hat{Y}_p- f(\hat{X}_p)\right) \geq \frac{k(\epsilon_{h,j}-\rho_1\gamma^h)}{2} \right)
	\stackrel{{(b)}}{\leq} \sum_{k=m+1}^{t} \exp\left(  -\frac{k\left(\epsilon_{h,j}-\rho_1\gamma^h \right)^2}{2\sigma^2}\right) \\
	&\leq t  \exp\left(  -\frac{m\left(\epsilon_{h,j}-\rho_1\gamma^h \right)^2}{2\sigma^2}\right) 
	\stackrel{{(c)}}{\leq} t\exp\left( -4\log T \right)=tT^{-4}\\
	\end{split}
	\end{equation}
In \eqref{equ: finalbound_E>f^star}, $(a)$ is due to a similar union bound in \eqref{equ:finalE_bound}(a); $(b)$ is obtained via the Hoeffding-Azuma inequality; $(c)$ is obtained via the substitution of $m\geq \frac{8\sigma^2 \log T}{\left(\epsilon_{h,j}-\rho_1\gamma^h \right)^2}$.

Finally, substituting \eqref{equ:m_value}, \eqref{equ:finalE_bound} and \eqref{equ: finalbound_E>f^star} into \eqref{equ:N_h_j_bound }, the upper bound  is given by
	\begin{equation}
	\begin{split}
	\mathbb{E}\left[	N_{h,j}(T) \right] 
	&\leq \frac{8\sigma^2 \log T}{\left(\epsilon_{h,j}-\rho_1\gamma^h \right)^2}+1+\sum_{t=m+1}^{T} \left(\sum_{k+1}^{t-1}t^{-3}+tT^{-4}\right)\\
	&\leq \frac{8\sigma^2 \log T}{\left(\epsilon_{h,j}-\rho_1\gamma^h \right)^2}+1+\sum_{t=1}^{T}\left(t^{-2}+T^{-3}\right)
	 \leq \frac{8\sigma^2 \log T}{\left(\epsilon_{h,j}-\rho_1\gamma^h \right)^2}+c
	\end{split}
	\end{equation}
	where $c$ is a constant. The last step is because $\sum_{t=1}^{T}t^{-2}$ is bounded. Hence, Lemma \ref{lemma:expected_N_bound} is proved.

\subsection{Proof of Theorem \ref{theorem: regret_bound}}\label{appendix:regret bound proof}

All nodes with depth $h$ can be divided into two subsets: ${\Phi}_h$ that denotes the set of all the $2\rho_1\gamma^h$-optimal nodes, and ${\Omega}_h$ that denotes the set of nodes whose parents belong to ${\Phi}_{h-1}$ while itself does not belong to ${\Phi}_h$. Let $H\geq 1$ be an integer whose value is determined later. With above definition, $\mathcal{T}$ can be divided into three subtrees: $\mathcal{T}_1$, $\mathcal{T}_2$ and $\mathcal{T}_3$. 
Let $\mathcal{T}_1$ contain ${\Phi}_H$ and its decedents. Let $\mathcal{T}_2$ include all the $2\rho_1\gamma^h$-optimal nodes at all the depths smaller than $H$, i.e., $\mathcal{T}_2=\bigcup\limits_{h=1}^{H-1}{\Phi}_h$. Let $\mathcal{T}_3$ include all the nodes in ${\Omega}_h$ at all the depths smaller than $H$, i.e., $\mathcal{T}_3=\bigcup\limits_{h=1}^{H}{\Omega}_h$. Hence the cumulative regret can be partitioned as
\begin{equation}\label{equ:total_regret}
\begin{split}
R^\pi\left({T}\right)
&=\mathbb{E}\left[R^\pi\left(\mathcal{T}_1\right)\right]+\mathbb{E}\left [R^\pi\left(\mathcal{T}_2\right)\right]+\mathbb{E}\left [R^\pi\left(\mathcal{T}_3\right)\right]
\end{split}
\end{equation}
where 
$$\mathbb{E}\left [R^\pi\left(\mathcal{T}_i\right)\right]=\mathbb{E}\left [\sum_{t=1}^{T}\left(f^\star-f\left(X_t\right)\right)  \mathbbm{1}_{\{ (H_t,J_t)\in \mathcal{T}_i \}}\right].$$
 Next, the regret analysis follows the idea of bounding the regret on each subtree separately.

\textbf{Step 1: Bounding the regret on $\mathcal{T}_1$}. 
As each node in ${\Phi}_H$ is $2\rho_1\gamma^H$-optimal, all the beams located in ${\Phi}_H$ are $4\rho_1\gamma^H$-optimal, i.e., $f^\star-f\left( X_t\right)\leq4\rho_1\gamma^H$, $X_t\in {\Phi}_H$. In addition, it is obvious that the number of nodes in subtree $\mathcal{T}_1$ is smaller than the time horizon, i.e., $| \mathcal{T}_1|\leq T$ where $|\cdot|$ represents the cardinality operator. Therefore, the regret on $\mathcal{T}_1$ is upper bounded by
\begin{equation}\label{eq:regret_bound_T1}
\mathbb{E}\left[R^\pi\left(\mathcal{T}_1\right)\right]\leq4\rho_1\gamma^H T.
\end{equation}

\textbf{Step 2: Bounding the regret on $\mathcal{T}_2$}. 
As $\mathcal{T}_2=\bigcup\limits_{h=1}^{H-1}{\Phi}_h$ and each beam in ${\Phi}_h$ is $4\rho_1\gamma^h$-optimal, the regret on $\mathcal{T}_2$ can be written as
$\mathbb{E}\left[R^\pi\left(\mathcal{T}_2\right)\right]\leq\sum_{h=1}^{H-1}4\rho_1\gamma^h |\Phi_h|.$
Based on the results in \cite{NIPS}, we have $|\Phi_h|\leq c_1\left(\rho_2 \gamma^h \right)^{-\kappa}$ where $\kappa=\frac{1}{\beta}-\frac{1}{\alpha}$. Specifically, $\alpha$ and $\beta$ are give in the weak Lipschitz assumption and the dissimilarity function, respectively.
The regret on $\mathcal{T}_2$ can be further bounded by 
\begin{equation}\label{equ:regret_bound_T2}
\begin{split}
\mathbb{E}\left[R^\pi\left(\mathcal{T}_2\right)\right]
&\leq \sum_{h=1}^{H-1}4\rho_1\gamma^h c_1\left(\rho_2 \gamma^h \right)^{-\kappa}
=4\rho_1c_1\rho_2^{-\kappa}\sum_{h=0}^{H-1}\gamma^{h(1-\kappa)}
\leq \frac{4\rho_1c_1\rho_2^{-\kappa}}{1-\gamma^{1-\kappa}}.
\end{split}
\end{equation}
From \eqref{equ:regret_bound_T2}, we can see that  $\mathbb{E}\left[R^\pi\left(\mathcal{T}_2\right)\right]$ is upper bounded by a constant as  $\mathcal{T}_2$ is a finite tree.

\textbf{Step 3: Bounding the regret on $\mathcal{T}_3$}. For each node in $\Omega_h$, its parents should be included by $\Phi_{h-1}$. Thus, all the beams in $\Omega_h$ are $4\rho_1\gamma^{h-1}$-optimal and the cardinality of $\Omega_h$ is smaller than $2|\Phi_{h-1}|$. Besides, with the results in Lemma \ref{lemma:expected_N_bound}, $\mathbb{E}\left[ N_{h,j}(t)\right] =\frac{8\sigma^2 \log t}{\left(\rho_1 \gamma^h\right)^2}+c $, for any $2\rho_1\gamma^{h-1}$-optimal nodes. Thus, the regret on $\mathcal{T}_3$ is given by
\begin{equation}\label{equ:regret_bound_T3}
\begin{split}
\mathbb{E}\left[R^\pi\left(\mathcal{T}_3\right)\right]
&\leq\sum_{h=1}^{H}4\rho_1\gamma^{h-1}2|\Phi_{h-1}| \mathbb{E}\left[ N_{h,j}(T) \right]
\leq 8\rho_1c_1\rho_2^{-\kappa} \sum_{h=1}^{H}\gamma^{(h-1)(1-\kappa)}\left(\frac{8\sigma^2 \log T}{\left(\rho_1 \gamma^h\right)^2}+c\right).\\
\end{split}
\end{equation}

Finally, substituting \eqref{eq:regret_bound_T1}, \eqref{equ:regret_bound_T2} and \eqref{equ:regret_bound_T3} into \eqref{equ:total_regret}, we have 
\begin{equation}\label{equ:regret_bigO}
\begin{split}
R^\pi\left({T}\right)
&\leq 4\rho_1\gamma^H T +\frac{4\rho_1c_1\rho_2^{-\kappa}}{1-\gamma^{1-\kappa}}+  8\rho_1c_1\rho_2^{-\kappa} \sum_{h=1}^{H}\gamma^{(h-1)(1-\kappa)}\left(\frac{8\sigma^2 \log T}{\left(\rho_1 \gamma^h\right)^2}+c\right)\\
&= O\left( \gamma^H T+\log T\gamma^{-H(1+\kappa)}  \right)
=O\left( T^{\frac{\kappa+1}{\kappa+2}}\left(\log T\right)^{\frac{1}{\kappa+2}} \right).
\end{split}
\end{equation}
The last step is obtained from setting $\gamma^H$ as the order of $\left({T}/{\log T}\right)^{-{1}/{(\kappa+2)}}$ \cite{NIPS}. {If the smoothness of the function is known, we can set $\alpha=\beta$ such that $\kappa=0$ \cite{NIPS}}. Hence, \eqref{equ:regret_bigO} can be rewritten as $O\left(\sqrt{T\log T}\right)$, and then the theorem is proved.

\bibliographystyle{IEEEtran}
\bibliography{security}

\end{document}